\documentclass[
reprint,
superscriptaddress,
amsmath,amssymb,
aps,
floatfix,
]{revtex4-2}

\usepackage{graphicx}
\usepackage{dcolumn}
\usepackage{bm}

\usepackage{amsmath}
\usepackage{amsthm}
\usepackage{subfiles}
\usepackage{physics}
\usepackage{float}
\usepackage{booktabs}
\usepackage{arydshln}
\usepackage{hyperref}
\usepackage{comment}
\usepackage{xr}

\newenvironment{vctr}
{\left[\begin{array}{c}}
{\end{array}\right]}

\newtheorem{theorem}{Theorem}

\begin{document}

\preprint{APS/123-QED}

\title{Multi-player conflict avoidance through entangled quantum walks}


\author{Honoka Shiratori}
\affiliation{Department of Information Physics and Computing, Graduate School of Information Science and Technology, The University of Tokyo, 7--3--1 Hongo, Bunkyo--ku, Tokyo 113--8656, Japan.}
\author{Tomoki Yamagami}
\email{tyamagami@mail.saitama-u.ac.jp}
\affiliation{Department of Information and Computer Sciences, Saitama University,
255 Shimo--okubo, Sakura--ku, Saitama City, Saitama 338--8570, Japan.}
\author{Etsuo Segawa}
\affiliation{Graduate School of Environment and Information Sciences,
Yokohama National University, 79--1 Tokiwadai, Hodogaya--ku, Yokohama, Kanagawa 240--8501, Japan.}
\author{\\Takatomo Mihana}
\affiliation{Department of Information Physics and Computing, Graduate School of Information Science and Technology, The University of Tokyo, 7--3--1 Hongo, Bunkyo--ku, Tokyo 113--8656, Japan.}
\author{Andr\'e R\"ohm}
\affiliation{Department of Information Physics and Computing, Graduate School of Information Science and Technology, The University of Tokyo, 7--3--1 Hongo, Bunkyo--ku, Tokyo 113--8656, Japan.}
\author{Ryoichi Horisaki}
\affiliation{Department of Information Physics and Computing, Graduate School of Information Science and Technology, The University of Tokyo, 7--3--1 Hongo, Bunkyo--ku, Tokyo 113--8656, Japan.}

\begin{abstract}
Quantum computing has the potential to solve complex problems faster and more efficiently than classical computing. It can achieve speedups by leveraging quantum phenomena like superposition, entanglement, and tunneling.
Quantum walks (QWs) form the foundation for many quantum algorithms. Unlike classical random walks, QWs exhibit quantum interference, leading to unique behaviors such as linear spreading and localization. These properties make QWs valuable for various applications, including universal computation, time series prediction, encryption, and quantum hash functions.  
One emerging application of QWs is decision making. Previous research has used QWs to model human decision processes and solve multi-armed bandit problems. This paper extends QWs to collective decision making, focusing on minimizing decision conflicts—cases where multiple agents choose the same option, leading to inefficiencies like traffic congestion or overloaded servers.  
Prior research using quantum interference has addressed two-player conflict avoidance but struggled with three-player scenarios. This paper proposes a novel method using QWs to entirely eliminate decision conflicts in three-player cases, demonstrating its effectiveness in collective decision making.

\end{abstract}

\maketitle



\section{Introduction}
The continuing increase in computational demand 
has attracted attention to new ways of computing~\cite{mehonic2022brain}.
In particular, quantum computers utilize tunneling, entanglement, and superposition and can potentially finish calculations in much shorter times and with less energy than traditional classical computing~\cite{meier2023energy}.
Exponential speedups attainable by quantum computing can make problems solvable that are impractical for classical ones to deal with.
Various quantum algorithms have been proposed in anticipation of future quantum computer hardware.
Famously, Grover's algorithm can extract an item from a randomly ordered list of $N$ items with only $O(\sqrt{N})$ iterations \cite{grover1996fast}, beating the fastest classical search algorithm \cite{bennett1997strengths}. 
Other examples are the Harrow-Hassidim-Lloyd algorithm that approximates the solution of simultaneous linear equations \cite{Harrow2009quantum}, and Shor's algorithm factoring an integer $N$ in time polynomial in the number of bits in $N$ \cite{Shor1994algorithms}.

The quantum walk (QW) is one basis for quantum algorithms and can be used for walking a graph \cite{douglas2009efficient}, finding target vertices \cite{magniez2007search}, and examining whether there is a way to get from node $s$ to node $t$ by following the edges of the graph ($s$-$t$ connectivity) \cite{kendon2011quantum}.
The QW is the quantum version of the classical random walk.
Whereas the classical random walker selects directions stochastically at each time step, and observers can always track its location, the quantum walker's location is unknown until measurement.
Due to quantum interference with itself, the probability distribution of a QWer is different from a classical random walker and features linear spreading and localization~\cite{konno2008quantum}.
Here, linear spreading describes the effect that the standard deviation of the probability distribution increases in proportion to time, i.e., it spreads faster than that of a classical random walker, whose standard deviation only grows as the square root of time.
Conversely, via localization, a specific position can retain a high probability no matter how long the walk continues for.

Various applications of QWs are known beyond finding graph properties.
The characteristic behaviors of QWs make them useful as a modeling tool to make quantum models of classical counterparts and design quantum automata.
For example, QWs can be a tool for universal computation \cite{childs2009universal}, predicting time series \cite{konno2019new}, encrypting images \cite{abd2019novel, abd2019encryption}, and realizing quantum hash functions \cite{yang2016quantum}.
One research field to which QWs have recently been applied is decision making.
Firstly, QWs have been used to model humans' decision-making processes \cite{martinez2016quantum, zhang2021quantum, busemeyer2006quantum, yan2022information}.
Furthermore, it can be used to find good options in exploitation-exploration dilemmas of multi-armed bandit problems \cite{yamagami2023bandit}.

In this paper, we apply QWs to collective decision-making, that is, decision problems involving multiple parties.
In particular, we present methods for how to avoid unwanted decision conflicts, i.e., multiple actors choosing the same option.
Decision conflicts model real-life scenarios of lost utility, such as when people choose the same option as their own to other's detriments. 
Examples include when too many cars go to the same road and cause traffic jams or when servers go down due to too many simultaneous access requests.

There has been some previous research on utilizing quantum effects for quantum computing by relying on photonics.
Amakasu et al. proposed a quantum interference system based on photons in which two players always avoid decision conflicts and choose different options \cite{amakasu}.
However, completely avoiding decision conflicts between three or more people with that quantum interference system has not been achieved in that setting.
In contrast, this paper applies QWs to collective decision making and details how to to avoid decision conflicts.
By manipulating the transition probabilities of the QWs the system avoids decision conflicts in three-player cases.

\section{QW on a 1D lattice or cycle}
This section gives a brief introduction to QWs and the relevant notation.
We focus on discrete-time QWs, where every node has inner states and their probability amplitudes evolve at each step.

\begin{figure}[t]
    \centering
    \includegraphics[width=0.45\textwidth]{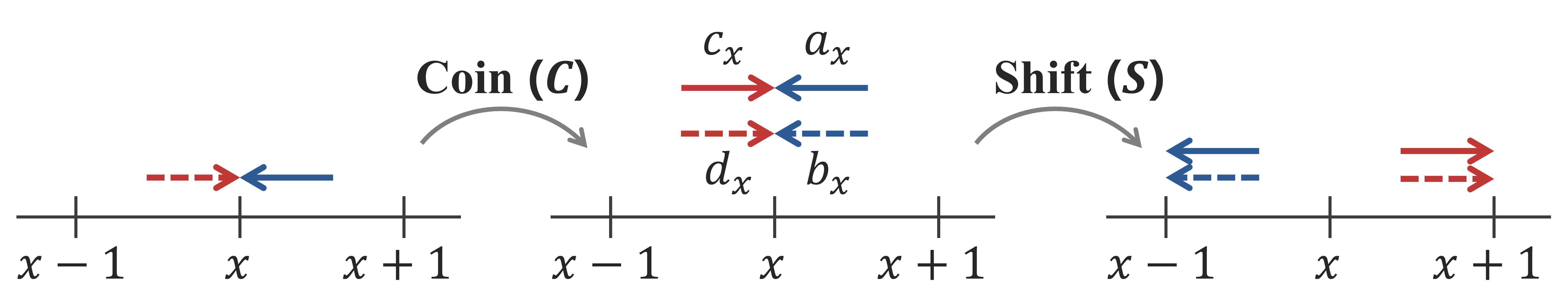}
    \caption{Schematic of QW on a 1D lattice. Each node has inner states: leftward and rightward. The coin matrix is applied to inner states, and new superposition of inner states is produced as a result (See Eq.~\eqref{eq:1d_coin}). The shift operator moves the leftward inner state to the left adjacent node and the rightward inner state to the right adjacent node  (See Eq.~\eqref{eq:1d_shift}).}
    \label{qw1}
\end{figure}

\subsection{Mathematics of Discrete QW on 1D lattice or cycle}\label{subsec:qw1d}

The state of a QW on a one-dimensional lattice is described by the superposition of states regarding the position, which is labeled by integers in order, and the two inner states, leftward ($L$) and rightward ($R$). Mathematically, a state $\Psi$ of the QW is a unit vector belonging to the compound Hilbert space
\begin{align}
    \mathcal{H} &= \mathcal{H}_{p}\otimes \mathcal{H}_{c} \notag \\
    &:= \operatorname{span}\{\ket{x}\otimes \ket{J}\,|\,x\in\mathbb{Z},\ J\in \{L,\,R\}\},
\end{align}
and thus, for any state $\Psi\in\mathcal{H}$, there exists a collection $(\alpha_{x,\,J}\,|\,x\in\mathbb{Z},\,J\in\{L,\,R\})$ satisfying
\begin{align}
    \Psi &= \sum_{x\in\mathbb{Z}}\sum_{J\in\{L,R\}}\alpha_{x,J}\ket{x}\otimes \ket{J} \notag \\
    &= \sum_{x\in\mathbb{Z}} \ket{x}\otimes (\alpha_{x,L}\ket{L} +\alpha_{x,R}\ket{R})
\end{align}
with $\sum_{x\in \mathbb{Z}}\sum_{J\in\{L,R\}}|\alpha_{x,J}|^2 = 1$. Herein, the absolute square of the probability amplitudes of $\ket{x}\otimes\ket{L}$ and $\ket{x}\otimes\ket{R}$ adds up to the chance of observing a walker at node $x$; that is, the probability with which the walker is found at $x\in\mathbb{Z}$ through measurement under the state $\Psi\in\mathcal{H}$ is denoted by
\begin{equation}
    p(x) = |\alpha_{x,L}|^2 + |\alpha_{x,R}|^2.
\end{equation}
The time evolution of a QW on a one-dimensional lattice is depicted in Fig.~\ref{qw1}.
The operator of the time evolution is composed of the shift and the coin operators. 
The sum of observation probabilities at all nodes, or the total of absolute squares of all inner states, must always equal unity, which requires that both operators are unitary.

The coin operator $C_1$ comprises local coin operators $C(x)$ at node $x\in\mathbb{Z}$. 
The local coin operator $C(x)$ converts the superposition of inner states at node $x$ at the current time step to that at the subsequent step.
Let vectors $[1,~0]^\top$ and $[0,~1]^\top$ be assigned to $L$ and $R$, respectively, i.e., $\ket{L}=[1,~0]^\top,~\ket{R}=[0,~1]^\top$, respectively. Herein, the superscript $^\top$ indicates the transpose. Then, denoting $C(x)$ as 
\begin{equation}
    C(x)=\left[\begin{array}{cc}
        a_x & b_x\\
        c_x & d_x
    \end{array}\right],
\end{equation}
the new inner states generated by the coin are
\begin{equation}
    \begin{split}
        C(x)\ket{L} =a_x\ket{L}+c_x\ket{R},\\
        C(x)\ket{R} =b_x\ket{L}+d_x\ket{R}.
    \end{split}
    \label{eq:1d_coin}
\end{equation}
The whole coin operator over $\mathcal{H}$ is defined as 
\begin{equation}
    C_1=\sum_{x} \ket{x}\bra{x} \otimes C(x).
\end{equation}
The shift operator $S_1$ shifts the leftward inner state at node $x$ to the left-adjacent one $x-1$ and the rightward inner state at node $x$ to the right-adjacent one $x+1$:
\begin{equation}
    \begin{split}
        S_1\ket{x}\otimes \ket{L} &= \ket{x-1}\otimes\ket{L},\\
        S_1\ket{x}\otimes\ket{R} &= \ket{x+1}\otimes\ket{R}.
    \end{split}
\end{equation}
That is, the operation of $S_1$ is described by
\begin{equation}
    S_1 = \sum_{x} \bigl(\ket{x-1}\bra{x} \otimes \ket{L}\bra{L} + \ket{x+1}\bra{x} \otimes \ket{R}\bra{R}\bigr).
    \label{eq:1d_shift}
\end{equation}
The time evolution of QWs can be realized by combining coin and shift.
The time evolution matrix of the QW is $W_1=S_1\cdot C_1$. That is, the evolving sequence of states $(\Psi^{(t)}\,|\,t = 0,\,1,\,\cdots)$ is described by
\begin{equation}
    \Psi^{(t+1)} = W_1\Psi^{(t)}.
\end{equation}

QWs on a cycle with $N$ nodes can be set up similarly except that the shift operator satisfies $1-1=N,~N+1=1$ at the endpoints $1,~N$. 
In this case, the QW state can be represented by a $2N$ dimensional vector $\Phi \in \mathbb{C}^{2N}$, which is spanned by vectors $\ket{1}\otimes\ket{L},~\ket{1}\otimes\ket{R},~\cdots,~\ket{N}\otimes\ket{L},~\ket{N}\otimes\ket{R}$.
It is the tensor product of a vector corresponding to the nodes and that corresponding to inner states, $\{\ket{1},\cdots, \ket{N}\} \otimes \{\ket{L}, \ket{R}\}$.

\subsection{Average probability distribution}
The average probability distribution is defined as
\begin{equation}
    \bar{p}_T(x)=\frac{1}{T} \sum_{t=0}^{T-1}p_t(x),
\end{equation}
where $p_t(x)$ is the probability of observing a walker at node $x$ at time $t$~\cite{portugal2013quantum}.
The average probability distribution is known to converge at $T\rightarrow \infty$ and is used in the analysis of QWs.
This paper shows the average probability distribution up to time $T=200$ obtained from numerical simulations.

\subsection{Relating QWs to decision making}
QWs on a 1D lattice or cycle can be utilized to model a player's decision making. 
``decision making" is a subset of reinforcement learning and describes the process by which an agent selects actions in a (usually Markov) environment to maximize a performance metric called the return. 
At each step, the agent observes its state, chooses an action, and receives a reward and a new state. 
Through repeated interactions, it refines its policy to balance exploration (testing unknown actions) with exploitation (using known successful actions), ultimately maximizing the expected return.
Ref. \cite{yamagami2023bandit} presents previous results of how QWs can be used for decision making, in the simplest general framework of the ``multi-armed bandit."

Each time the agent has to make a decision, it starts a quantum walk from some defined initial position. 
It lets the quantum walk evolve under the time evolution operator $W_1$, until a predefined end time.
Then, it measures the resulting position of the quantum walker.
Each spatial position $x$ is associated with an action, with the observed outcome being chosen as the action for that agent in this round.
Then, for the next round, the process repeats (after potentially adjusting the coin operators to incorporate new knowledge about the reward environment).
Thus, we can utilize the inherent uncertainty of the QW as the basis for exploration, while fine-tuning of the coin matrices can be used to achieve exploitation, e.g., by achieving localization.

When such decision-making problems contain multiple players, avoiding decision conflicts often becomes a vital aspect.
When resources are limited and options cannot be shared, players need to not only explore and exploit to maximize their rewards, but also strive to avoid choosing the same option at the same time.
Quantum walks offer a novel way to implement this coordination through the fundamental properties of their time evolution. 

There are several ways to extend decision making using QWs to multiple players. 
Fundamentally, we will let each player have their own ``quantum walker," whose measurement outcome determines that players action. 
Thus, if player one observes position $\ket{x}$ and player two observers $\ket{y}$, our goal is to minimize or ideally completely avoid the cases of $x=y$ in each individual round. 
At the same time, we would like to preserve each players freedom of choice when considered over all rounds, i.e., each player should be able to reach any position and thus able to choose any action in principle.

We will show how QWs can be used to solve the coordination problem of multi-player decision making and discuss the advantages and disadvantages of various approaches in the rest of this paper.

\section{Decision making of two players on separate 1D lattices}

\subsection{Entanglement of the initial state}
\label{section:initialstate}

The most straightforward way to extend decision making to multiple players is by simply using tensor products of QWs on the 1D lattice.
For example, two players' decision making can be modeled by the tensor product of two QWs.
Let player A and player B each have their own QW on a separate 1D lattices.
Their collective decision making can be $(x,\,y)$ when player A's quantum walker is observed at node $x$ and player B's at node $y$.
The entire system can be represented by the tensor product of two separate QWs on 1D lattices.
However, introducing entanglement of these two QWs is necessary to avoid conflicts.
As the simplest approach, we will explore the properties of the entanglement of the initial state and ability to influence collective decision making.

Denoting the states of QWs for players A and B on the 1D lattice at time $t$ as $\Phi_A^{(t)}$ and $\Phi_B^{(t)}$, respectively, their time evolution can be described by
\begin{equation}
    \Phi_A^{(t)} = W_A^t \Phi_A^{(0)}, ~\Phi_B^{(t)} = W_B^t \Phi_B^{(0)},
\end{equation}
respectively.
Here $W_A$ and $W_B$ are the time evolution matrices of QWs of players A and B, respectively, and time is represented by $t$:
\begin{equation}
    \begin{split}
        W_A = \sum_{x} \bigl(\ket{x-1}\bra{x} \otimes \ket{L}_x\bra{L}_x C_A(x) + \\
        \ket{x+1}\bra{x} \otimes \ket{R}_x\bra{R}_x C_A(x) \bigr),
    \end{split}
\end{equation}
\begin{equation}
    \begin{split}
        W_B = \sum_{y} \bigl(\ket{y-1}\bra{y} \otimes \ket{L}_y\bra{L}_y C_B(y) + \\
        \ket{y+1}\bra{y} \otimes \ket{R}_y\bra{R}_y C_B(y) \bigr).
    \end{split}
\end{equation}
Player A operates on spatial positions $\ket{x}$ and Player B on the separate 1D lattice $\ket{y}$.
The tensor product of these two QWs is represented by
\begin{align}
    \begin{split}
        \Phi_{all}^{(t)} &= \Phi_A^{(t)} \otimes \Phi_B^{(t)} = W_A^t \Phi_A^{(0)} \otimes W_B^t \Phi_B^{(0)} \\
        &=  (W_A \otimes W_B)^t (\Phi_A^{(0)} \otimes \Phi_B^{(0)}).
    \end{split}
\end{align}
The operator $W_A \otimes W_B$ can be regarded as the time evolution matrix of the entire state.

Then, consider entangling the initial state.
For example, a vector 
\begin{equation}
    \frac{1}{\sqrt{2}}(\Psi_1 \otimes \Psi_2 + \Psi_2 \otimes \Psi_1)
\end{equation}
is entangled.
When it is adopted as the initial state, the time evolution of the entire state is obtained as follows:
\begin{align}
    \begin{split}
        \Phi_{all}^{(t)} =\,& (W_A \otimes W_B)^t \frac{1}{\sqrt{2}}(\Psi_1 \otimes \Psi_2 + \Psi_2 \otimes \Psi_1) \\
        =\,& \frac{1}{\sqrt{2}} (W_A \otimes W_B)^t ~ \Psi_1 \otimes \Psi_2 \mathrel{+}\\
        &\frac{1}{\sqrt{2}} (W_A \otimes W_B)^t ~ \Psi_2 \otimes \Psi_1 \\
        =\,& \frac{1}{\sqrt{2}} \left( W_A^t \Psi_1 \otimes W_B^t \Psi_2 + W_A^t \Psi_2 \otimes W_B^t \Psi_1  \right).
    \end{split}
\end{align}
\subsection{Partial avoidance through fermion}
\label{sec:fermion}

The avoidance of selection conflicts is related to a fermion-like property of the joint probability amplitude.
Let $\Phi_{all}$ be the probability amplitude of the joint system. 
If it meets the condition of
\begin{equation}
    \Phi_{all}(x,\,y) = -\Phi_{all}(y,\,x),
    \label{eq:fermi_node}
\end{equation}
then the conflict is completely avoided because $\Phi_{all}(x,\,x) = 0.$
Eq.~\eqref{eq:fermi_node} is realized under a certain time evolution matrix $W$.
The following theorem describes the condition that $W$ must satisfy to realize Eq.~\eqref{eq:fermi_node}.

\begin{theorem}
    \label{the:fermion}
    Let $W$ be the time evolution matrix of the QW and $U_\sigma$ be a swap operator: $(U_\sigma \Phi_{all})(x,\,y)=\Phi_{all}(y,\,x)$.
    Then the following holds.
    \begin{equation}
    \begin{split}
        &W = U_\sigma W U_\sigma^{-1}, ~U_\sigma\Phi_{all}^{(0)} = -\Phi_{all}^{(0)} \\
        &\Leftrightarrow~ U_\sigma \Phi_{all}^{(t)} = -\Phi_{all}^{(t)}, \\
        &\hspace{20pt} \text{i.e., }~ \Phi_{all}^{(t)}(x,\,y) = -\Phi_{all}^{(t)}(y,\,x),~^\forall t \ge 0.
    \end{split}
    \end{equation}
\end{theorem}

Therefore, based on this theorem, if a $W$ satisfying 
\begin{equation}
\begin{split}
    &W = U_\sigma W U_\sigma^{-1}, \\
    &\text{i.e.,}~ (W\Phi_{all})(x,\,y) = (U_\sigma W U_\sigma^{-1}\Phi_{all})(x,\,y) 
    \label{eq:w_condition}
\end{split}
\end{equation}
could be found, the joint probability amplitude has fermionic properties (see Eq.~\eqref{eq:fermi_node}), and complete conflict avoidance is possible because of the Theorem~\ref{the:fermion}.
However, it can be shown that such unitary matrix $W$ cannot exist, as we prove in the supplementary material.

Therefore, we need to consider the fermion of nodes and inner states, i.e., while we may not prevent both walkers from being at the same spatial location $x=y$, we will at least try to prevent them from being in the same exact location and inner state.
Let the inner states at positions $x$ and $y$ be represented by $\mu$ and $\nu \in \{L,\,R\}$, respectively.
The condition that the joint probability amplitude is fermionic in nodes and inner states requires that
\begin{equation}
    \Phi((x,\,\mu),(y,\,\nu)) = -\Phi((y,\,\nu),(x,\,\mu)),
    \label{eq:fermi}
\end{equation}
which leads to $\Phi((x,\,\mu), (x,\,\mu)) = 0.$
This implies that the probability amplitudes of the pairs of the same nodes and inner states are zero, which is a path towards partial conflict avoidance.

\begin{theorem}
    When the initial state is the entangled vector 
    \begin{equation}
        \Psi_1 \otimes \Psi_2 - \Psi_2 \otimes \Psi_1,
        \label{eq:ini}
    \end{equation}
    and the two operators $W_A$ and $W_B$ are commutative for the tensor product $W_A \otimes W_B = W_B \otimes W_A$, then the state vector at any time satisfies Eq.~\eqref{eq:fermi}.
\end{theorem}

\begin{proof}
    First, the initial state Eq.~\eqref{eq:ini} satisfies Eq.~\eqref{eq:fermi} because 
    \begin{equation}
    \begin{split}
        &(\Psi_1 \otimes \Psi_2 - \Psi_2 \otimes \Psi_1)((x,\,\mu),(y,\,\nu)) \\
        &= \Psi_1(x,\,\mu) \otimes \Psi_2(y,\,\nu) - \Psi_2(x,\,\mu) \otimes \Psi_1(y,\,\nu) \\
        &= - \left( \Psi_2(x,\,\mu) \otimes \Psi_1(y,\,\nu) - \Psi_1(x,\,\mu) \otimes \Psi_2(y,\,\nu) \right) \\
        &= - (\Psi_1 \otimes \Psi_2 - \Psi_2 \otimes \Psi_1)((y,\,\nu), (x,\,\mu)).
    \end{split}
    \end{equation}
    Second, at time $t$, because $W_A \otimes W_B = W_B \otimes W_A$,
    \begin{equation}
        \begin{split}
            &(W_A \otimes W_B)^t (\Psi_1 \otimes \Psi_2 - \Psi_2 \otimes \Psi_1)\bigr)((x,\,\mu),(y,\,\nu))\\
            &= - \bigl( W_B^t \Psi_2  (x,\,\mu) \otimes W_A^t \Psi_1 (y,\,\nu) \\
            &\hspace{24pt} - W_B^t \Psi_1  (x,\,\mu) \otimes W_A^t \Psi_2 (y,\,\nu) \bigr) \\
            &= - \bigl((W_A \otimes W_B)^t (\Psi_1 \otimes \Psi_2 - \Psi_2 \otimes \Psi_1)\bigr)((y,\,\nu),(x,\,\mu)).
        \end{split}
    \end{equation}
    Therefore, the state vector satisfies Eq.~\eqref{eq:fermi} at any time.
\end{proof}

\begin{figure}[b]
   \centering
    \includegraphics[width=0.45\textwidth]{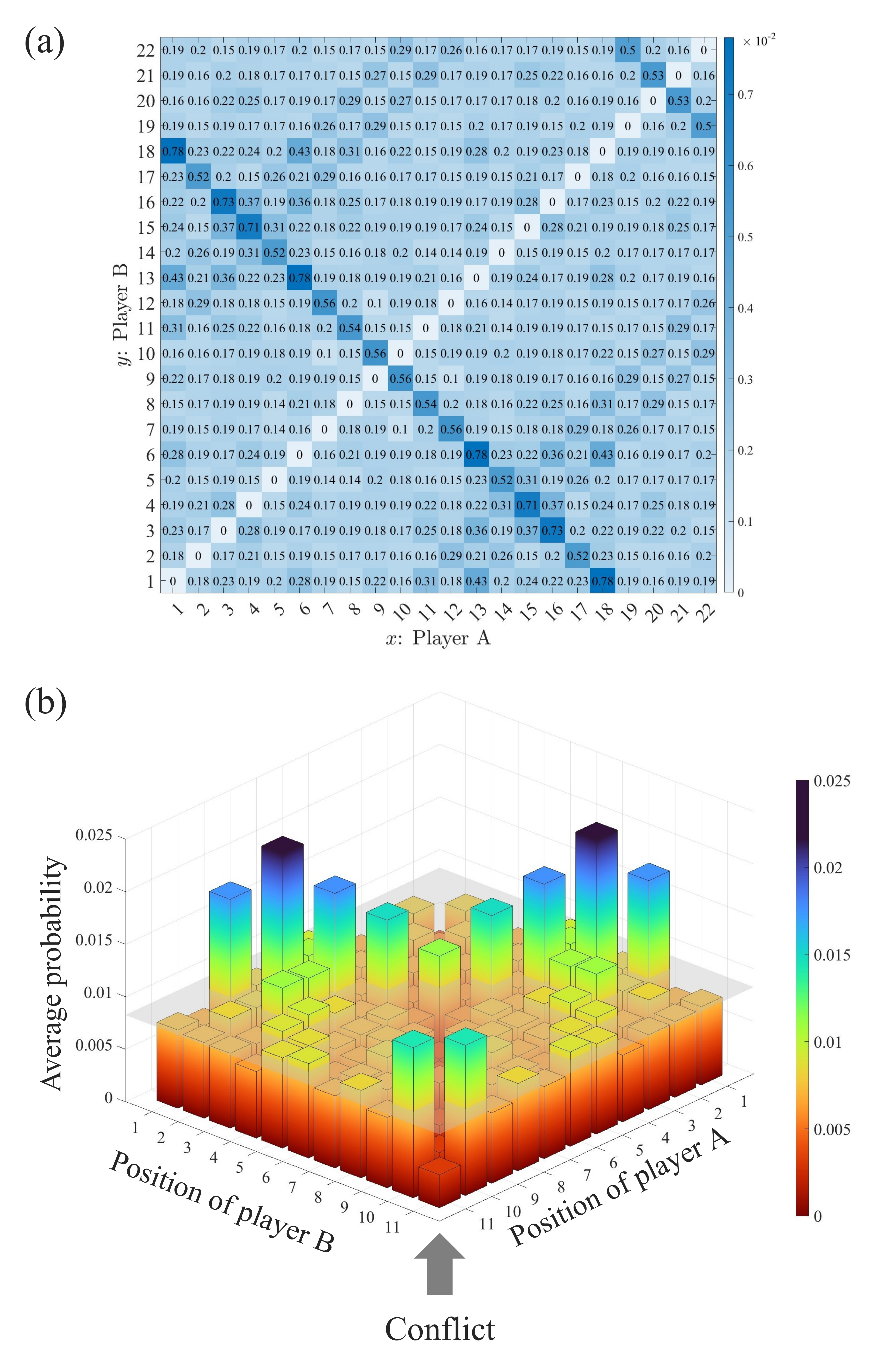}
   \caption{Numerical simulation of partial avoidance through fermion of nodes. (a) Average probability distribution of inner states. The pairs of the same node and the same inner state, such as $(1L, ~1L)$, have zero probability amplitude by realizing the state shown in Eq.~\eqref{eq:fermi}. (b) Average probability distribution of nodes. The probability of observing the walker at conflict nodes, i.e., conflict probability, is not necessarily smaller than the uniform distribution.}
   \label{fig:spsi}
\end{figure}

For example, when $\Phi_1$ and $\Phi_2$ have two elements, i.e., the cycles have only one node, Eq.~\eqref{eq:ini} becomes
\begin{equation}
    \left[ \begin{array}{c} a \\ b \end{array} \right] \otimes \left[ \begin{array}{c} c \\ d \end{array} \right] - \left[ \begin{array}{c} c \\ d \end{array} \right] \otimes \left[ \begin{array}{c} a \\ b \end{array} \right] = \left[ \begin{array}{c} 0 \\ ad-bc \\ bc-ad \\ 0 \end{array} \right].
\end{equation}
This implies that the probability amplitudes of $((1,\,L), (1,\,L))$ and $((1,\,R), (1,\,R))$ are zero.
Similarly, when $\Phi_1$ and $\Phi_2$ have four elements, i.e., the cycles have two nodes, Eq.~\eqref{eq:ini} becomes
\begin{equation}
\begin{split}
    &\left[ \begin{array}{c} a_1 \\ b_1 \\ a_2 \\ b_2 \end{array} \right] \otimes 
    \left[ \begin{array}{c} c_1 \\ d_1 \\ c_2 \\ d_2 \end{array} \right] - 
    \left[ \begin{array}{c} c_1 \\ d_1 \\ c_2 \\ d_2 \end{array} \right] \otimes 
    \left[ \begin{array}{c} a_1 \\ b_1 \\ a_2 \\ b_2 \end{array} \right] \\
    &= 
    [ \begin{array}{ccccc} 0 & a_1d_1-b_1c_1 & a_1c_2 - a_2c_1 & a_1d_2-b_2c_1   \end{array}  \\
    &\hspace{15pt} \begin{array}{cccc}  b_1c_1-a_1d_1 & 0 & b_1c_2-a_2d_1 & b_1d_2-b_2d_1   \end{array}  \\
    &\hspace{15pt} \begin{array}{cccc} a_2c_1-a_1c_2 & a_2d_1-b_1c_2 & 0 & a_2d_2-b_2c_2    \end{array} \\
    &\hspace{15pt} \begin{array}{cccc}  b_2c_1-a_1d_2 & b_2d_1-b_1d_2 & b_2c_2-a_2d_2 & 0  \end{array}]^\top.
\end{split}
\end{equation}
This means the probability amplitudes of $((1,\,L), (1,\,L))$, $((1,\,R), (1,\,R))$, $((2,\,L), (2,\,L))$ and $((2,\,R), (2,\,R))$ are zero. 
Therefore, the pairs of the same node and the same inner states have zero probability amplitudes in the vectors which have the same form as Eq.~\eqref{eq:ini}.

The easiest way to realize the commutative condition of $W_A$ and $W_B$ is to make them equal.
Therefore, the numerical simulation below adopts $W_A=W_B=W$.

Figures~\ref{fig:spsi}(a) and (b) show an example for partial conflict avoidance through fermion, where the initial state is
\begin{equation}
\begin{split}
    &\Phi_{all}^{(0)} = \frac{1}{\sqrt{2}}\left( \Psi_A\otimes \Psi_B - \Psi_B \otimes \Psi_A  \right), \\
    &\Psi_A = \ket{2}\otimes\frac{1}{\sqrt{2}} \begin{vctr}1\\ i\end{vctr}, ~\Psi_B = \ket{8}\otimes\frac{1}{\sqrt{2}} \begin{vctr}i\\ 1\end{vctr}, 
    \label{eq:init_fermi}
\end{split}
\end{equation}
Here, $\ket{x}$ is formalized by a column vector with $N$ elements whose $x$-th one is $1$ and others are $0$.
Here, $e_x$ represents a column vector with $N$ elements whose $x$-th one is $1$ and others are $0$.
Each node on the 1D cycle has a $2 \times 2$ matrix coin
\begin{equation}
    \frac{1}{\sqrt{2}}\left[\begin{array}{cc}
        1 & 1 \\
        -1 & 1
    \end{array} \right] .
\end{equation}
The average probability distribution for all possible combinations of nodes and inner states for $N=11$ is displayed in Fig.~\ref{fig:spsi}(a).
For $x=1,\,2,\,\cdots,\,11$, the $(2x-1)$-th and $2x$-th rows and columns  in Fig.~\ref{fig:spsi}(a) correspond to $\ket{x,\,L}$ and $\ket{x,\,R}$ for players A and B, respectively.
Successfully, all of the right-down diagonal elements that correspond to the pairs of the same node and inner state are zero.
However, the two quantum walkers can still be found in the same position, as long as they do not share the same inner state, e.g., $\Phi((x,\,L), (x,\,R)) \neq 0.$

The 3D bar graph, shown in Fig.~\ref{fig:spsi}(b), illustrates this problem. It displays the average probability for every possible combination of positions. 
This is the supposed measurement outcome used in the decision making of each agent.
Crucially, the probabilities at conflict nodes $(1,\,1),~(2,\,2),\cdots (11,\,11)$ (shown by bars along the arrow in Fig.~\ref{fig:spsi}(b)) are non-zero. 
The conflict nodes, on average, have a lower average probability, indicating that conflicts are partially averted.
However, in places the average probability can reach the height of the uniform distribution displayed by the gray plane. 
Further tuning of the coin matrices could likely improve this result, but it can never reach full conflict avoidance, due to the impossibility of a unitary operator $W$ fulfilling Theorem~\ref{the:fermion} (proven in the supplementary material).

\subsection{Partial avoidance through fermion and mirrored state}

We extend the method of the previous subsection by utilizing ``mirrored states."

The definition of a mirrored state is introduced first, before the detailed explanation of its effect on the quantum walk probability distribution.
Let a vector $\Phi$ represent the QW state on a 1D cycle with three nodes for example.
The vector $\Phi$ has six probability amplitudes because there are three nodes and two inner states
\begin{equation}
    \Phi = [a_1,~b_1,~a_2,~b_2,~a_3,~b_3]^\top.
\end{equation}
Probability amplitudes $a_1$ and $b_1$ correspond to $\ket{1,\,L}$ and $\ket{1,\,R}$, respectively, $a_2$ and $b_2$ correspond to $\ket{2,\,L}$ and $\ket{2,\,R}$, respectively, and $a_3$ and $b_3$ correspond to $\ket{3,\,L}$ and $\ket{3,\,R}$, respectively.
The mirrored state $\Hat{\Phi}$ has probability amplitudes mirrored around a certain position $x^*$ from $\Phi$.
For example, if $x^*=2$, then the mirrored state is
\begin{equation}
    \Hat{\Phi} = [b_3,~a_3,~b_2,~a_2,~b_1,~a_1]^\top.
\end{equation}

The following state consisting of mirror states makes the central position $x^*$ conflict-free:
\begin{equation}
    \begin{split}
        \Phi_{all}^{(t)} = &\Phi_A^{(t)} \otimes \Phi_B^{(t)} - \Phi_B^{(t)} \otimes \Phi_A^{(t)} + \\
        &\Hat{\Phi}_A^{(t)} \otimes \Hat{\Phi}_B^{(t)} - \Hat{\Phi}_B^{(t)} \otimes \Hat{\Phi}_A^{(t)}.
    \end{split}
\end{equation}
To confirm, the inner states at position $x^*$ is
\begin{align}\begin{split}
    &\Phi_{all}^{(t)}(x^*,\,x^*)  \\
    &= \Phi_A(x^*)\otimes \Phi_B(x^*) - \Phi_B(x^*) \otimes \Phi_A(x^*) \\
    &\hspace{10pt}+ \Hat{\Phi}_A(x^*) \otimes \Hat{\Phi}_B(x^*) - \Hat{\Phi}_B(x^*) \otimes \Hat{\Phi}_A(x^*) \\
    &= \begin{vctr}a\\b \end{vctr} \otimes \begin{vctr}c\\d \end{vctr} - \begin{vctr}c\\d \end{vctr} \otimes \begin{vctr}a\\b \end{vctr} \\
    &\hspace{10pt}+ \begin{vctr}b\\a \end{vctr} \otimes \begin{vctr}d\\c \end{vctr} - \begin{vctr}d\\c \end{vctr} \otimes \begin{vctr}b\\a \end{vctr} = \begin{vctr}0\\0\\0\\0 \end{vctr}.
    \end{split}
    \label{eq:mirrored}
\end{align}
Therefore, perfect conflict avoidance covering all conflict nodes is impossible, but avoiding one position is feasible by the initial state entanglement.

\begin{figure}[t]
   \centering
    \includegraphics[width=0.45\textwidth]{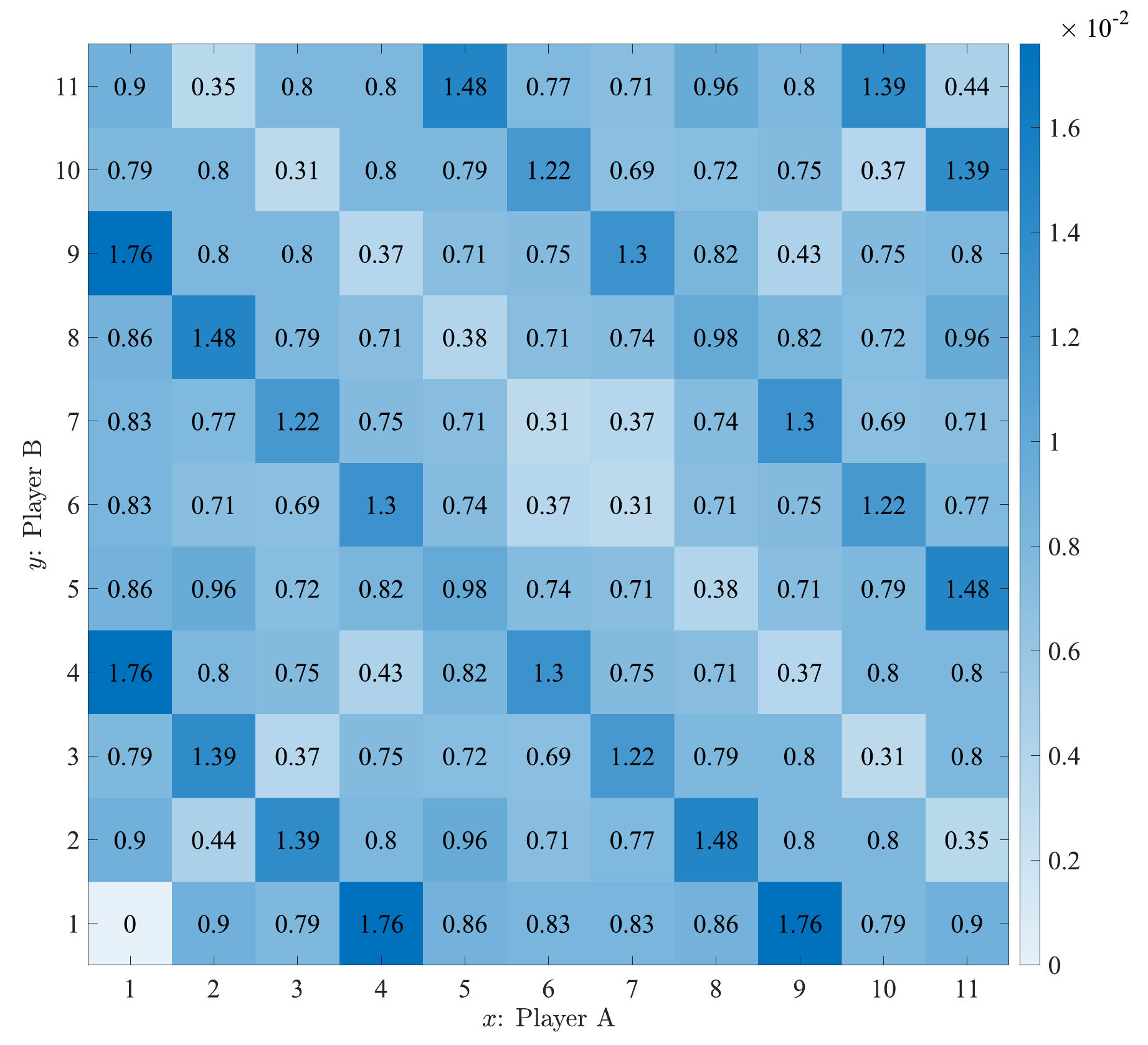}
   \caption{Average probability obtained for partial avoidance through fermion and mirrored state. The center of the mirrored state $x^*$ is at position $x=y=1$. The conflict probability at node $(1, ~1)$ is resultingly zero, and therefore, choice $1$ is made conflict-free as intended (See Eq.~\eqref{eq:mirrored}).}
   \label{fig:one}
\end{figure}

The simulation results are shown in Fig.~\ref{fig:one}, where $x^*=1$, $N=11$.
The initial state is
\begin{equation}
\begin{split}
    &\Phi_{all}^{(0)} = \frac{1}{2} (a_0 \otimes b_0 - b_0 \otimes a_0  + \Hat{a}_0 \otimes \Hat{b}_0 - \Hat{b}_0 \otimes \Hat{a}_0), \\
    &a_0 = \ket{2} \otimes\frac{1}{\sqrt{2}}\begin{vctr}1\\ i\end{vctr},~b_0 = \ket{8} \otimes\frac{1}{\sqrt{2}}\begin{vctr}i\\ 1\end{vctr}, \\
    &\Hat{a}_0 = \ket{N}\otimes\frac{1}{\sqrt{2}}\begin{vctr}i\\ 1\end{vctr},~~\Hat{b}_0 = \ket{5}\otimes\frac{1}{\sqrt{2}}\begin{vctr}1\\ i\end{vctr}.
\end{split}
\end{equation}
The coin of each node on the 1D cycle is
\begin{equation}
    \frac{1}{2}\left[
    \begin{array}{cc}
        1+i & 1-i \\
        1-i & 1+i
    \end{array}
    \right].
\end{equation}
The average probability distribution for positions is shown in Fig.~\ref{fig:one}, and the central position $(1,\,1)$ is conflict-free as intended. 
Thus, by engineering our initial state to have certain spatial symmetries, we can avoid conflicts at one spatial position (but not all at the same time).

Overall, the entanglement-based approach for the initial states has some advantages, i.e., it is conceptually a straightforward extensions of 1D QW decision making and it requires only a careful preparation of the initial state, but is otherwise quite liberal for the choice of time-evolution operators $W$.
However, as we have shown, complete conflict avoidance is impossible in this framework. 
To obtain conflict-free decision making for two and three players, we need to entangle the time-evolution operators, i.e., the coin operators of the system. 
This yields a new kind of topology for the QW and is analyzed in the rest of this paper.

\section{decision making of two players by QW on a 2D lattice or torus}

The initial state's entanglement has the potential to either make at one option free from conflicts, or do reduce the conflict rate by engineering the probability amplitudes to be zero for the same location and inner state.
Ideally, we want to make all options conflict-free. 
A natural way to move beyond merely entangling two one-dimensional QWs is to merge the two spatial dimensions into a single two-dimensional lattice and treat the entire walk as one four-dimensional coin system. 
In other words, rather than having separate 1D walkers whose initial states may be entangled, we now give a single walker a 2D position space and a 4D coin space.
This allows the coin operator itself to become intrinsically entangled—no longer decomposable into separate factors for each dimension—thus providing a richer framework for controlling the walk and, in our setting, enabling more powerful mechanisms to avoid conflicts in multi-player decision-making.

\subsection{Theory of QWs on a 2D lattice or torus}

We first introduce the general theory of QWs on a 2D lattice. 
Here, four inner states are taken into consideration: $\ket{L}_x\ket{L}_y$, $\ket{L}_x\ket{R}_y$, $\ket{R}_x\ket{L}_y$, and $\ket{R}_x\ket{R}_y$.
The coin $C(x,\,y)$ at node $(x,\,y)$ is a $4 \times 4$ matrix.
The coin acting on the whole state is
\begin{equation}
    C_2 =\sum_x \sum_y \ket{x,\,y}\bra{x,\,y} \otimes C(x,\,y). 
    \label{eq:2d_coin}
\end{equation}
The shift operator is
\begin{equation}\begin{split}
    S_2 = &\sum_x \sum_y \bigl( \ket{x-1,\,y-1}\bra{x,\,y}\otimes \ket{L}_x\ket{L}_y\bra{L}_x\bra{L}_y \\
    &+ \ket{x-1,\,y+1}\bra{x,\,y} \otimes \ket{L}_x\ket{R}_y\bra{L}_x\bra{R}_y \\
    &+ \ket{x+1,\,y-1}\bra{x,\,y} \otimes \ket{R}_x\ket{L}_y\bra{R}_x\bra{L}_y \\
    &+ \ket{x+1,\,y+1}\bra{x,\,y} \otimes \ket{R}_x\ket{R}_y \bra{R}_x\bra{R}_y \bigr),
\end{split}\end{equation}
and the time evolution operator of the QW is 
\begin{align}
    \begin{split}
        &W_2 = S_2 \cdot C_2 \\
    &= \sum_{x,~y} \bigl( \ket{x-1,\,y-1}\bra{x,\,y}\otimes \ket{L}_x\ket{L}_y\bra{L}_x\bra{L}_yC(x,\,y)  \\
    &+ \ket{x-1,\,y+1}\bra{x,\,y} \otimes \ket{L}_x\ket{R}_y\bra{L}_x\bra{R}_y C(x,\,y) \\
    &+ \ket{x+1,\,y-1}\bra{x,\,y} \otimes \ket{R}_x\ket{L}_y\bra{R}_x\bra{L}_y C(x,\,y) \\
    &+ \ket{x+1,\,y+1}\bra{x,\,y} \otimes \ket{R}_x\ket{R}_y \bra{R}_x\bra{R}_y C(x,\,y) \bigr).
    \end{split}
\end{align} 

A QW on a 2D torus can be defined by adjusting the shift at the endpoints appropriately, like in the 1D case.
The state vector $\Phi$ of the QW can be spanned by the tensor product of a position vector and an inner state vector $\{\ket{1,\,1},\,\ket{1,\,2},\,\cdots,\,\ket{N,\,N}\}\otimes\{\ket{LL},\ket{LR},\ket{RL},\ket{RR}\}$.

\subsection{Relating QWs on a 2D lattice or torus to decision making}

The main aim of this paper is to present a new method for avoiding conflicts utilizing QWs.
Unlike in the case of two separate 1D walkers, we are now considering a single walker on 2D plane.
This quantum walker walks on a 2D torus, where each node has an associated $2^2\times 2^2=4 \times 4$ unitary matrix coin.
For each decision to be made by both agents, they collectively start one joint walker, who evolves as a 2D QW for a certain time.
Then, when the walker is observed at node $(i,\,j)$, the players A and B choose options $i$ and $j$, respectively.
Our goal remains to investigates whether it is possible to avoid decision conflicts using QWs.
Because players' decisions correspond to nodes where the quantum walker is observed, the conflict avoidance question now maps to the problem of preventing the quantum walker being observed at nodes $(1,\,1),\,(2,\,2),\,\cdots,\,(N,\,N)$.
Therefore, in the 2D lattice, we shall call the nodes $(1,\,1),\,(2,\,2),\,\cdots,\,(N,\,N)$ the ``conflict nodes," and the nodes that connect to the conflict nodes in the subsequent time step the ``border nodes."
Figure~\ref{fig:decision} shows our basic setup for the 2D case.

\begin{figure}[t]
   \centering
    \includegraphics[width=0.45\textwidth]{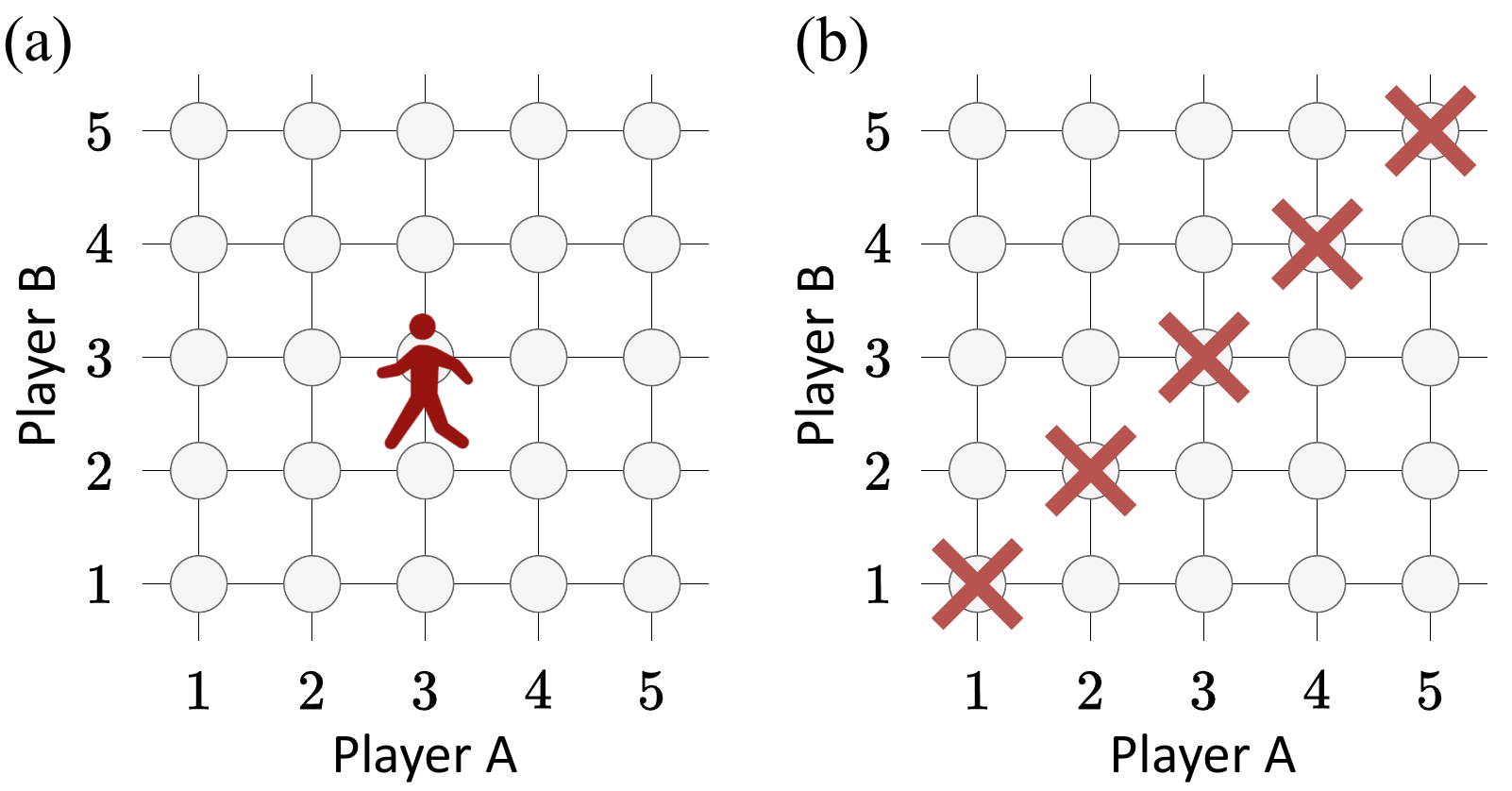}
   \caption{The quantum walker walks on a 2D lattice. The node where the walker is observed is considered as the collective decision the players make. (a) For example, players A and B choose options $3$ and $2$ when the walker is observed at node $(3,\,2)$. (b) The decision conflict correspond to cases where the walker is observed at nodes $(1,\,1),\,(2,\,2),\,\cdots, (N,\,N)$. The purpose of this paper is to avoid such cases.}
   \label{fig:decision}
\end{figure}


\begin{figure}[t]
   \centering
    \includegraphics[width=0.4\textwidth]{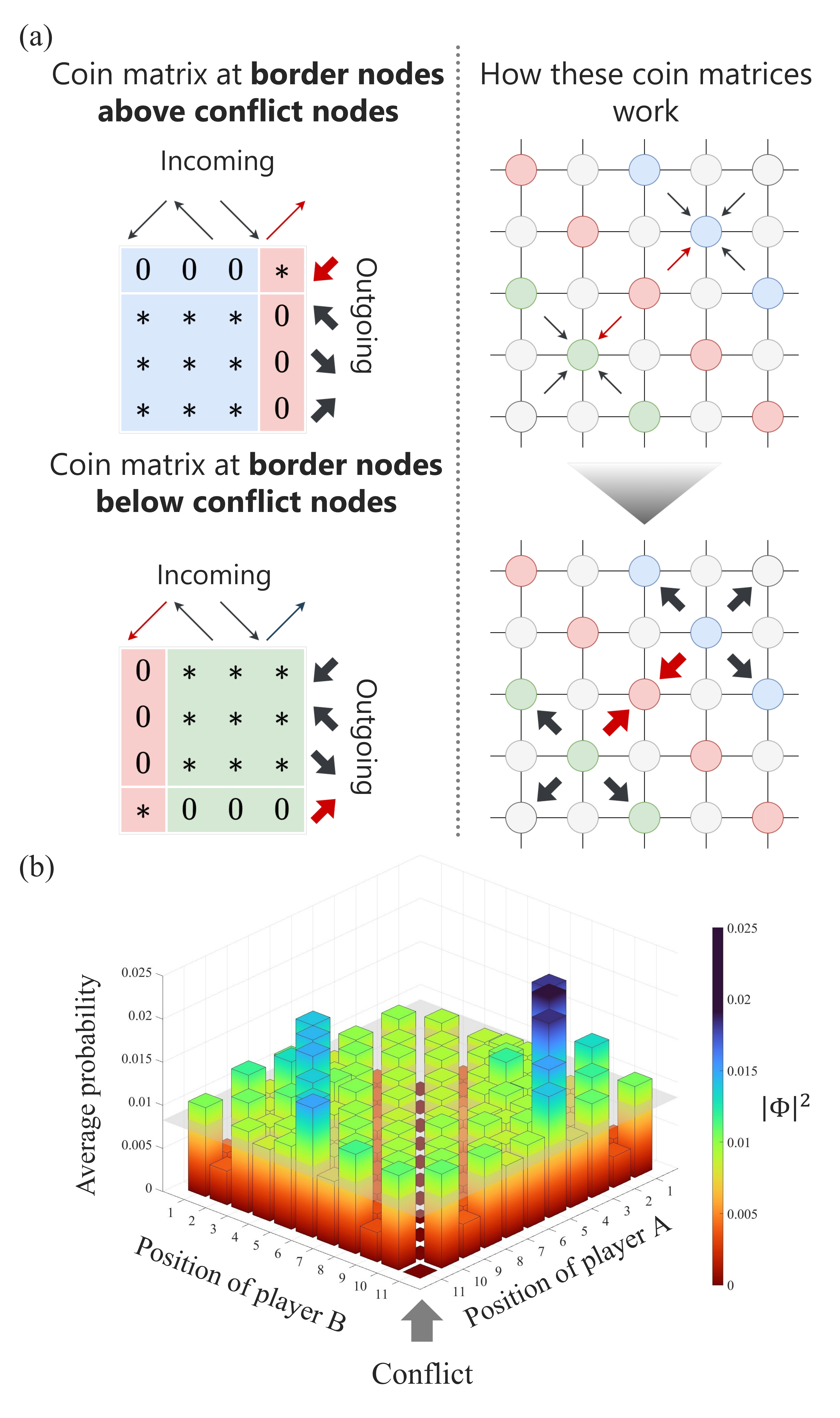}
   \caption{(a) Sketch of the coin operators (see Eq.~\eqref{eq:2d_coin}) to avoid decision conflicts for 2-player decision-making. Coin operators $C(x,\,y)$ are represented as a $4 \times 4$ matrix. Each node $(x,\,y)$ has four inner states, shown by arrows for indicating incoming and outgoing probability amplitudes. Coins at border nodes are entangled in such a way that probability amplitude flows from conflict nodes are isolated and reflected. (b) Average probability distribution when border nodes have entangled coins shown in (a). The conflict nodes are never reached by the walker as intended. Concrete coins used in the simulation are shown in Eqs.~\eqref{eq:coinabove}--\eqref{eq:coinother}.}
   \label{fig:reflect}
\end{figure}

\subsection{Entanglement of the coins: reflection}
\label{section:reflection}

We will now employ the additional freedom of entangled coins for conflict avoidance.
Thus, we need to assign to each node $(x,\,y)$ its $4\times 4$ coin matrix, where $x$ and $y$ denote both the position on the 2D lattice, as well as the actions later chosen by players A and B, respectively.
Since these coins no longer need to consist of tensor products of $2 \times 2$ coins, the term ``entanglement" is employed, and the entire system cannot be easily decomposed down into tensor product of two QWs on 1D lattices.

Complete conflict-avoidance is possible in this case. 
The main idea needed is to reflect the quantum walker at border nodes so that it can never reach conflict nodes in the first place. 

We visually represent how to solve the conflict avoidance problem in Fig.~\ref{fig:reflect}(a). 
Here, the coin operators are represented as $4\times4$ matrices, where elements marked as $*$ may take arbitrary values, under the constraint that $C(x,\,y)$ remains unitary. 
Figure~\ref{fig:reflect}(a) shows the coin matrices on the border nodes $(i,\,j)$ with $j = i \pm 2$.
Note that due to the walker having to move in both dimensions at every step, nodes at $(i,\,j)$ with $j=i\pm1$ are not directly connected to conflict nodes.

The coins change the inner states incoming from non-conflict nodes (shown by black arrows in Fig.~\ref{fig:reflect}(a)) and direct them back toward the non-conflict nodes.
This procedure prevents a quantum random walker on this 2D grid from ever reaching a conflict node, as long as the initial state was not a conflict node.
To meet the coin's unitary constraint, the inner states coming from conflict nodes (shown by red arrows in Fig.~\ref{fig:reflect}(a)) must be reflected back to those going to conflict nodes.
This also suggests that the walker is always observed at conflict nodes if it starts at conflict nodes.
In essence, by choosing the coin operators in this way, we are decomposing the graph into subgraphs from which a walker can never escape. 
This fact will become more relevant and important to study in the case of $3$ players discussed in Sec.~\ref{section:threeplayers}.

We numerically simulated this case.
The initial state of the simulation is $\ket{2,\,8} \otimes \frac{1}{2}[i,1,-1,i]^\top$, indicating that the walker starts from a position $(2,8)$ with the initial inner state $\frac{1}{2}(i\ket{LL} +\ket{LR} -\ket{RL} +i\ket{RR})$.
Note, that this initial position can be arbitrary as long as it is not a conflict node.

For $N=11$ options and $2$ players, the simulated average probability distribution of each node is displayed in Fig.~\ref{fig:reflect}(b).
In this numerical simulation, 
the coin matrix of the border nodes above conflict nodes (displayed by blue nodes in Fig.~\ref{fig:reflect}) is
\begin{equation}
    \left[ \begin{array}{cccc}
       0 & 0 & 0 & 1 \\
       -\frac{i}{\sqrt{2}}  & \frac{1}{2} & \frac{i}{2} & 0 \\
       \frac{1}{\sqrt{2}} & -\frac{i}{2} & \frac{1}{2} & 0 \\
       0 & \frac{1}{\sqrt{2}} & \frac{i}{\sqrt{2}} & 0
    \end{array} \right].
    \label{eq:coinabove}
\end{equation}
The coin matrix of the border nodes below conflict nodes (displayed by green nodes in Fig.~\ref{fig:reflect}) is
\begin{equation}
    \left[\begin{array}{cccc}
       0 & -\frac{i}{\sqrt{2}}  & \frac{1}{2} & \frac{i}{2} \\
       0 & \frac{1}{\sqrt{2}} & -\frac{i}{2} & \frac{1}{2} \\
       0 & 0 & \frac{1}{\sqrt{2}} & \frac{i}{\sqrt{2}}  \\
       1 & 0 & 0 & 0
    \end{array} \right].
\end{equation}
The coin matrix of the other nodes is
\begin{equation}
    \frac{1}{\sqrt{2}}\left[\begin{array}{cc}
        1 & 1 \\
        -1 & 1
    \end{array} \right] \otimes  \frac{1}{\sqrt{2}}\left[\begin{array}{cc}
        1 & 1 \\
        -1 & 1
    \end{array} \right].
    \label{eq:coinother}
\end{equation}
The conflict nodes on the line indicated by the arrow in Fig.~\ref{fig:reflect} have average probability of zero, showing that conflicts are completely avoided.
Because any walker that starts at a non-conflict node never goes to conflict nodes, we do not need to give conflict nodes special coins. 
Similarly, the choice of coins inside the bulk of non-border and non-conflict nodes can be engineered to suit the needs of the decision-making problem, while remaining conflict-free.
We have therefore shown via this example, that it is possible to use 2D quantum walks for avoiding decision-conflicts.
\label{section:2d}

\section{decision making of three players by QW on a 3D lattice or torus}
\begin{figure}[t]
    \centering
    \includegraphics[width=0.2\textwidth]{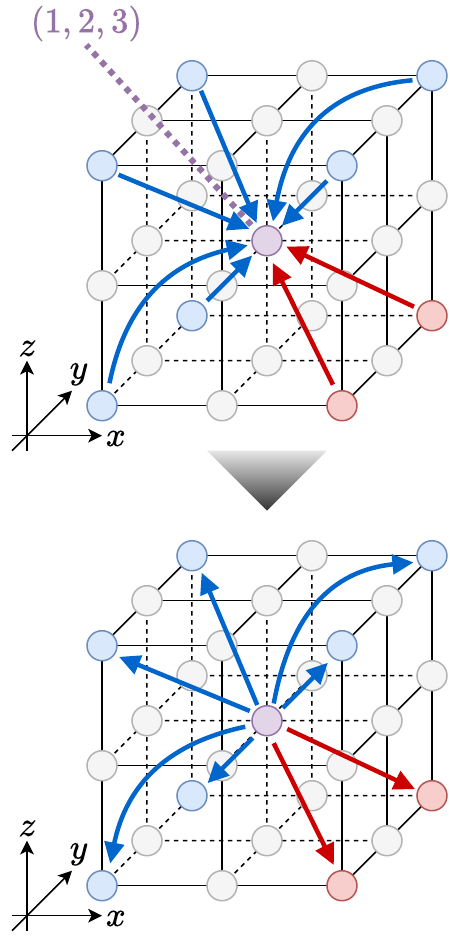}
    \caption{Inner states (a) before and (b) after being converted by the coin of node $(1,\,2,\,3)$. Blue arrows are from or to non-conflict nodes and red ones are from or to conflict nodes. Blue arrows are converted to blue ones and red ones are converted to red ones. This design makes conflict nodes have zero input when the walker starts at non-conflict nodes, and is realized by placing zeros properly in the coin (See Eq.~\eqref{eq:3p_coin_pattern}). 
    }
    \label{fig:reflection_image}
\end{figure}

\begin{figure}[t]
   \centering
    \includegraphics[width=0.45\textwidth]{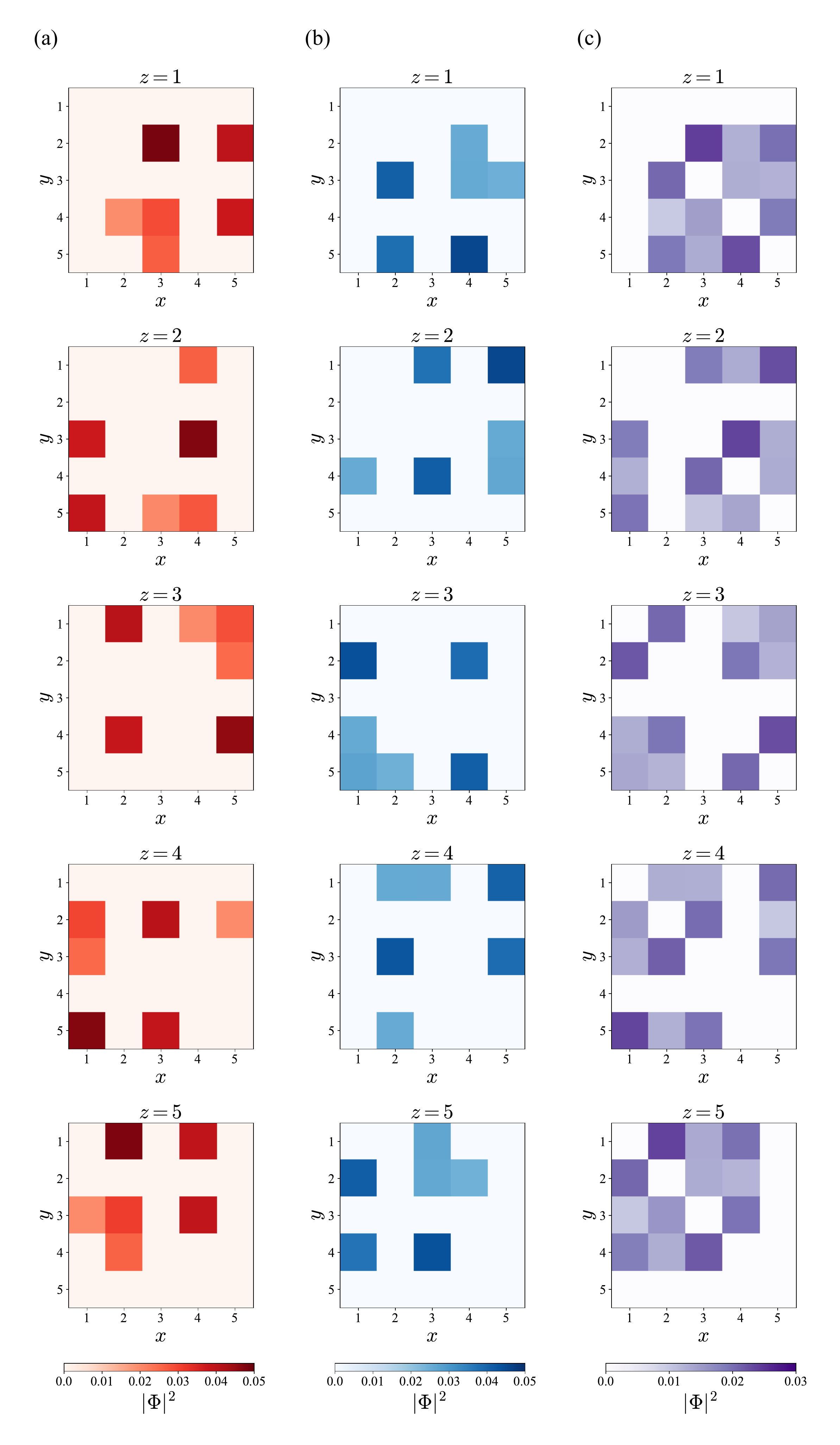}
   \caption{Average probability distribution after $200$ steps when the initial position is (a) $(1,\,2,\,3)$ and (b) $(3,\,2,\,1)$ when inner states are reflected at border nodes, as shown in Fig.~\ref{fig:reflection_image}.
   These colormaps are cross-sections with $z$ fixed, and they show probabilities at all nodes.
   Only half of the non-conflict nodes have non-zero average probabilities in (a) and (b) because nodes are divided into subnetworks.
   (c) Average probability distribution after $200$ steps when nodes $(1,\,2,\,3)$ and $(3,\,2,\,1)$ have probability amplitude of $1/\sqrt{2}$. All non-conflict nodes can be approached in this case thanks to the subnetwork analysis.}
   \label{fig:3p_avd}
\end{figure}


\subsection{First simulation for three players and its problem}
While the method introduced in the previous section can work for two players, there are additional complexities that only arise for larger player numbers. 
In the following, we generalize the method based on entangled coins introduced in Sec.~\ref{section:reflection} to three-player cases. 
Now, players A, B, and C choose options $i$, $j$, and $k$, respectively, derived from the measurement of a quantum walker moving on a 3D torus observed at node $(i,\,j,\,k)$.
Importantly, in the three-player case, the nature and amount of conflict nodes changes.
The decision conflict occurs not only when the walker is observed at nodes $(i,\,i,\,i)$ for $i=1,\,2,\,\cdots,\,N$ but also at nodes $(i,\,i,\,j)$ with permutation for $i\neq j$ and $i,\,j=1,\,2,\,\cdots,\,N$.
Therefore, avoiding decision conflicts in three-player cases is more difficult than in two-player cases.

In three-player cases, the coin of each node is a $2^3 \times 2^3 = 8 \times 8$ matrix.
The notation of the coin's rows and columns are $LLL$, $LLR$, $LRL$, $LRR$, $RLL$, $RLR$, $RRL$, and $RRR$.
 
As in section~\ref{section:reflection}, the quantum walker should be stopped at border nodes so that it cannot reach conflict nodes.
This is realized by correctly redirecting the walker.

As an example, let the number of options $N=5$, i.e. the number of nodes on the cycle, be five.
Then, the border node at $(1,\,2,\,3)$ is a border node, as the quantum walker can reach at least one conflict node within one step of the time evolution.
In particular, it can reach the conflict node at $(2,\,j,\,2)$.
To avoid conflicts, we will assigned the coin matrix
\begin{equation}
    \left[\begin{array}{cccccccc}
        * & 0 & * & 0 & * & * & * & * \\
        * & 0 & * & 0 & * & * & * & * \\
        * & 0 & * & 0 & * & * & * & * \\
        * & 0 & * & 0 & * & * & * & * \\
        0 & * & 0 & * & 0 & 0 & 0 & 0 \\
        * & 0 & * & 0 & * & * & * & * \\
        0 & * & 0 & * & 0 & 0 & 0 & 0 \\
        * & 0 & * & 0 & * & * & * & * 
    \end{array}\right],
    \label{eq:3p_coin_pattern}
\end{equation}
where $*$ means any complex number as long as the coin is unitary.

We illustrate the basic behavior in Fig.~\ref{fig:reflection_image}, now for the 3D case.
Inner states from non-conflict nodes (blue arrows in Fig.~\ref{fig:reflection_image}(a)) are converted to inner states which will go to non-conflict nodes in the following step (blue arrows in Fig.~\ref{fig:reflection_image}(b)).
Similarly, inner states from conflict nodes (red arrows in Fig.~\ref{fig:reflection_image}(a)) are turned into to inner states to conflict nodes (red arrows in Fig.~\ref{fig:reflection_image}(b)).
That is, the inputs to border nodes from non-conflict nodes are reflected to non-conflict nodes so that the walker cannot reach conflict nodes.

This coin matrix can be modified to have two diagonal blocks by permutation of rows and columns
\begin{equation}
    \left[\begin{array}{cccccccc}
        * & * & 0 & 0 & 0 & 0 & 0 & 0 \\
        * & * & 0 & 0 & 0 & 0 & 0 & 0 \\
        0 & 0 & * & * & * & * & * & * \\
        0 & 0 & * & * & * & * & * & * \\
        0 & 0 & * & * & * & * & * & * \\
        0 & 0 & * & * & * & * & * & * \\
        0 & 0 & * & * & * & * & * & * \\
        0 & 0 & * & * & * & * & * & * 
    \end{array}\right].
\end{equation}
Therefore, the two blocks should be unitary submatrices for the coin in Eq.~\eqref{eq:3p_coin_pattern} to be unitary.
The size of blocks differs according to the number of adjacent conflict nodes. 
This is different from the two player case, where the number of adjacent conflict nodes was always $0$ or $1$. 
With three players, some nodes have more than one adjacent conflict node.
The possible sizes of blocks are $2,~4,~6,$ and $8$.

We will now use a simple numerical simulation for exploring the dynamics of the conflict-avoiding QW with 3 players as modelled by a QW on a 3-torus with $N=5$.
In our numerical simulations we will use the four following submatrices for the unitary submatrices of the coin operators:
\begin{align}
    &U_2 = \frac{1}{\sqrt{2}}\left[ \begin{array}{cc}
        1 & 1 \\
        1 & -1
    \end{array} \right],\\
    &U_4 = \frac{1}{2}\left[ \begin{array}{rrrr}
        1 & 1  & 1 & 1  \\
        1 & -1 & 1 & -1   \\
        1 & 1 & -1 & -1  \\
        1 & -1 & -1 & 1  \\
    \end{array} \right],\\
    &U_6 = \frac{1}{2}\left[ \begin{array}{cccc}
        -\sqrt{2}i & 1 & i \\
        \sqrt{2} & -i & 1 \\
        0 & \sqrt{2} & -\sqrt{2}i
    \end{array} \right] \otimes \frac{1}{\sqrt{2}}\left[ \begin{array}{cc}
        1 & 1 \\
        1 & -1
    \end{array} \right],\\
    &U_8 = \frac{1}{2\sqrt{2}}\left[ \begin{array}{rrrrrrrr}
        1 & 1  & 1 & 1 & 1 & 1 & 1 & 1   \\
        1 & -1 & 1 & -1 & 1 & -1 & 1 & -1   \\
        1 & 1 & -1 & -1 & 1 & 1 & -1 & -1 \\
        1 & -1 & -1 & 1 & 1 & -1 & -1 & 1 \\
        1 & 1 & 1 & 1 & -1 & -1 & -1 & -1 \\
        1 & -1 & 1 & -1 & -1 & 1 & -1 & 1\\
        1 & 1  & -1 & -1 & -1 & -1 & 1 & 1\\
        1 & -1 & -1 & 1 & -1 & 1 & 1 & -1
    \end{array} \right].
\end{align}

Figure~\ref{fig:3p_avd}(a) shows the resultant average probability distribution when the quantum walker starts at node $(1,\,2,\,3)$.
Here, we represent the positions $(x,\,y,\,z)$ as five separate 2D slices, with the obtained average probability distribution indicated in red.

We can see that the method as generalized to the three player case does indeed work. 
All conflict nodes such as $(1,\,1,\,1)$ or $(1,\,1,\,2)$ have zero probability, because the walker never reaches the conflict nodes.

However, an additional problem becomes immediately obvious: There are $5 \times 4 \times 3 = 60$ non-conflict nodes, whereas only $30$ nodes have non-zero average probability in Fig.~\ref{fig:3p_avd}(a).
Conversely, when the walker starts at node $(3,\,2,\,1)$, the other half of non-conflict nodes have non-zero probability (see  Fig.~\ref{fig:3p_avd}(b)).
Nodes possessing non-zero probability in Fig.~\ref{fig:3p_avd}(a) have zero probability in Fig.~\ref{fig:3p_avd}(b) and vice versa.

This problem arises due to the effective network structure of the quantum walk.
The walker can move from non-conflict nodes to non-conflict nodes and from conflict nodes to conflict nodes.
However, it cannot move from non-conflict nodes to conflict nodes nor from conflict nodes to non-conflict nodes.
When the number of players is three, the probability of the conflict nodes is much higher compared to the case of two players.
As a result, the network decomposes into several subnetworks.
While this result does indeed avoid conflicts, it also violates the requirements for decision-making with multiple parties.
In particular, each player should have the possibility to choose the same options, at least in theory. 
Moreover, all collective choices that are non-conflict should be allowed. 
We do not achieve this in this case.

A way to overcome this problem is to make the starting position of the quantum walk distributed over several nodes. 
In practice, we would like to avoid requiring a complex procedure, and limit ourselves to as few starting positions as possible. 
In order to understand how to implement the decision making effectively, we need to understand the structure of the subnetworks that the system decomposes into, which we will be studying in the following.

\subsection{Network decomposition}
\label{section:net_dec}

\renewcommand{\arraystretch}{1.25}
\begin{table*}[t]
    \centering
    \caption{Number of nodes in each subnetwork. Some sizes are only conjectured based on the simulations in Fig.~\ref{fig:networks} in the supplementary material, while for some cases the sizes are strictly proven in Theorems~\ref{thm:odd} and \ref{thm:even}. SN stands for the subnetwork. Nodes are decomposed into multiple subnetworks when inner states are reflected, as shown in Fig.~\ref{fig:reflection_image}.}
    \label{tab:deduction}
    \begin{tabular}{lll}
        \toprule
         & When $N$ is odd & When $N$ is even \\
        \midrule
        \midrule
        \#Total nodes & $N^3$ & $N^3$ \\
        \midrule
        Parity & $N^3$ & $N^3/4$ \\
        \midrule
        \#Non-conflict nodes & $N(N-1)(N-2)$ & $N(N-1)(N-2)$ \\
         & SN1: $N(N-1)(N-2)/2$ & SN1: $N(N-2)(N-4)/8$ \\
        \#Nodes in each & SN2: $N(N-1)(N-2)/2$ & SN2: $N(N-2)(N-4)/8$ \\
        subnetwork  &  & SN3: $N^3/4 - N^2/2$ \\
         &  & SN4: $N^3/4 - N^2/2$ \\
         &  & SN5: $N^3/4 - N^2/2$ \\
        \midrule
        \#Conflict nodes & $3N^2 - 2N$ & $3N^2 - 2N$ \\
        \#Nodes in each & SN3: $3N^2 - 2N$ & SN6: $3N^2/2 - 2N$ \\
        subnetwork  &  & SN7: $N^2/2$ \\
         &  & SN8: $N^2/2$ \\
         &  & SN9: $N^2/2$ \\
        \bottomrule
    \end{tabular}
\end{table*}
\renewcommand{\arraystretch}{1.0}

\begin{figure}[b]
   \centering
    \includegraphics[width=0.45\textwidth]{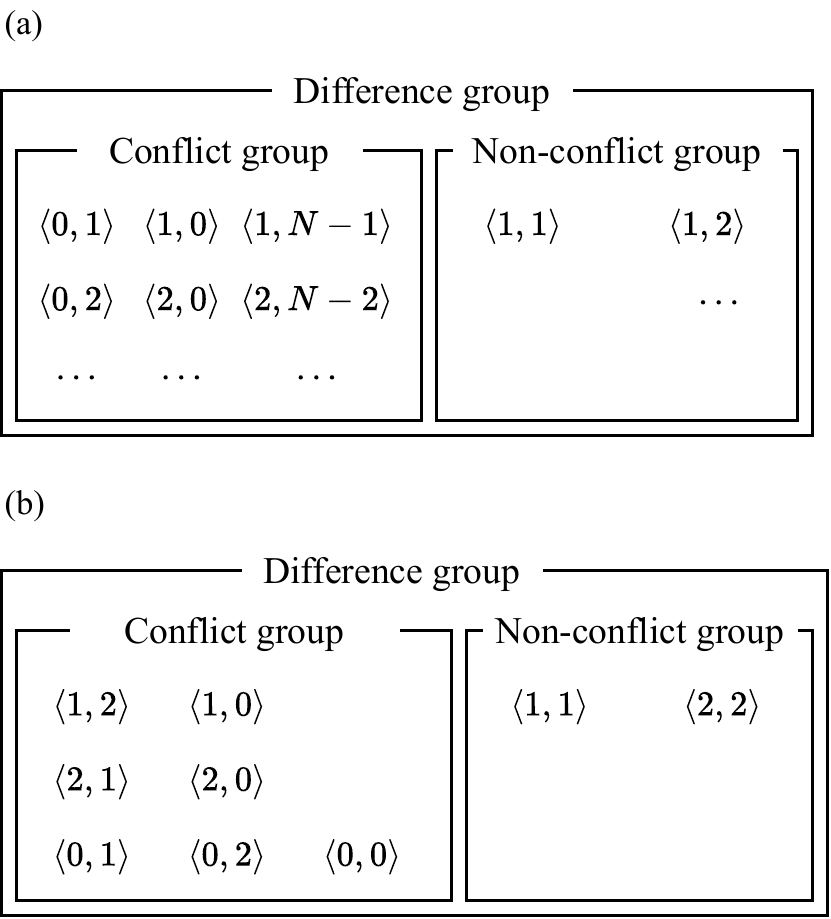}
   \caption{Examples of conflict and non-conflict groups (a) for general $N$ and (b) when $N=3$. Each difference group is either a conflict group or a non-conflict group. The difference group is a means of efficiently classifying nodes. The difference group $\langle l, m \rangle$ be the set of all nodes $(i,\,j,\,k)$ with $l = (j-i )~\text{mod}~N$ and $m = (k-j )~\text{mod}~N$.}
   \label{fig:ben}
\end{figure}

\begin{figure}[b]
   \centering
    \includegraphics[width=0.45\textwidth]{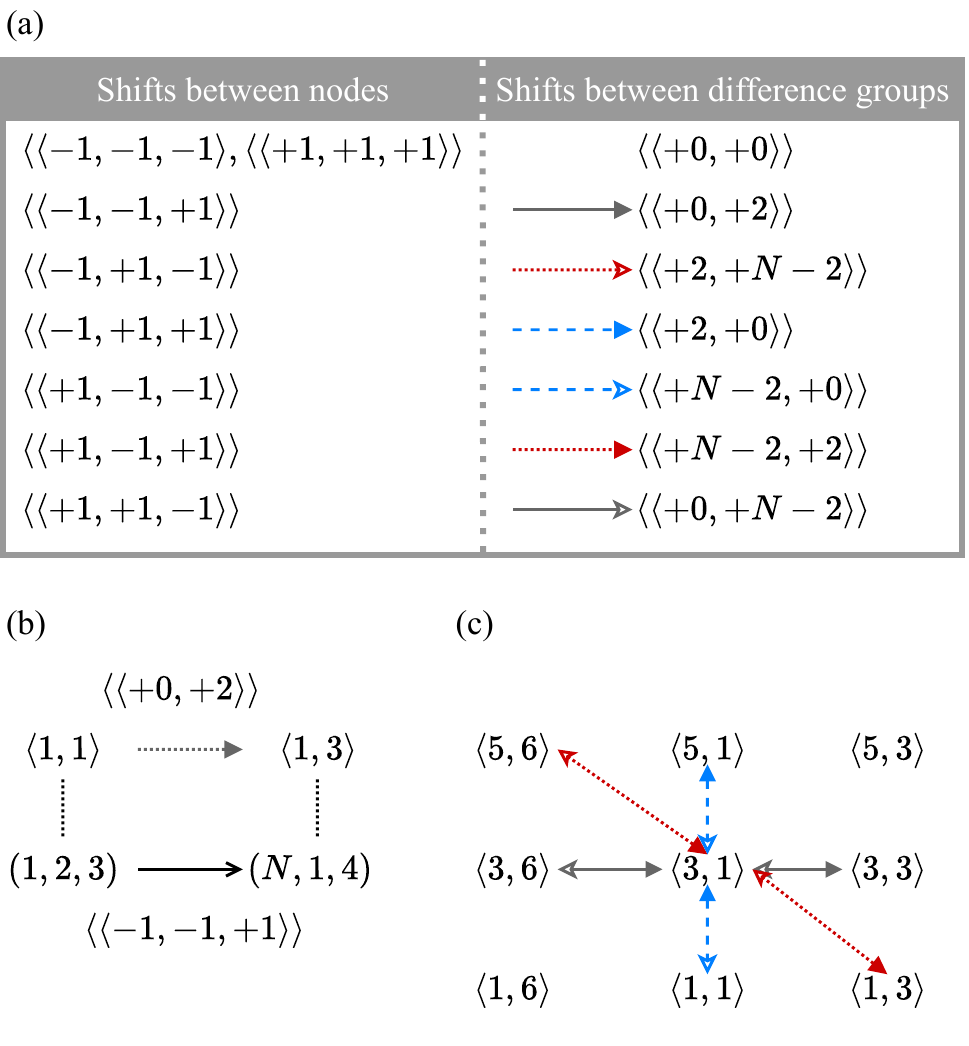}
   \caption{(a) Possible shifts between nodes and corresponding shifts between difference groups. There are $2^3=8$ shifts between nodes in total. They are mapped to seven possible shifts between difference groups.
   (b) Example of shifts between nodes and difference groups. (c) Example of possible shifts from difference group $\langle 3, 1\rangle$ when $N=7$. We need to consider six directions.}
   \label{fig:shifts}
\end{figure}

\begin{figure*}[t]
   \centering
    \includegraphics[width=0.8\textwidth]{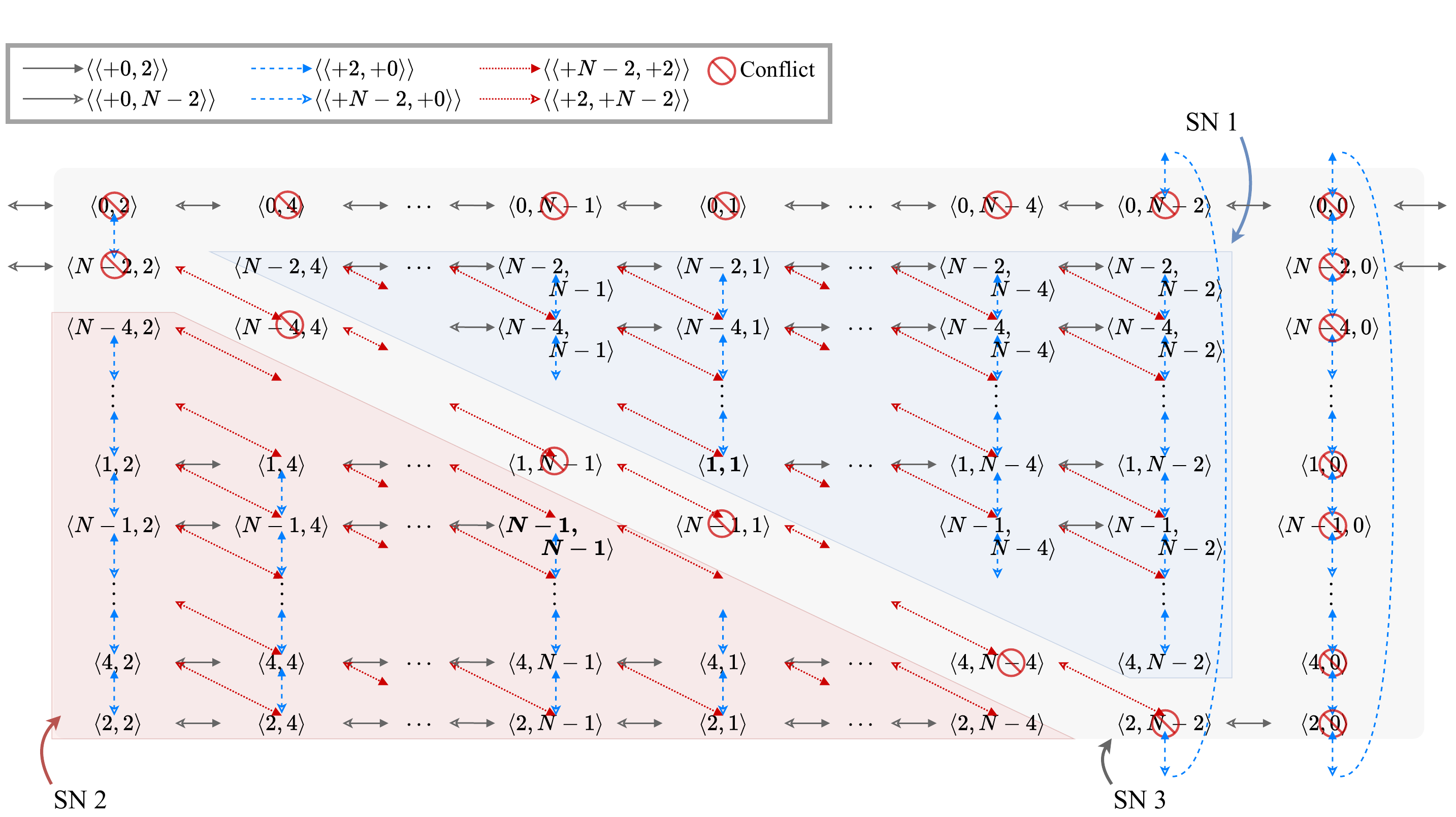}
   \caption{Transition between difference groups when $N$ is an any odd number. SN stands for the subnetwork. Note, that the 2D grid is ordered in a way where the difference group $\langle 1,1 \rangle$ is not next to $\langle 1, 2 \rangle$, as these difference groups are not connected by transitions.}
   \label{fig:odd_general}
\end{figure*}

The network of the quantum walker for three-players and $N$ options moves on the surface of a $3$-torus.
When the coins are chosen such that a walker starting on a non-conflict node cannot reach conflict nodes, this automatically decomposes the network into subnetworks.
The number and size of these subnetworks depends on the number of options $N$. 

To deduce the sizes of the subnetworks, we reduce the dimension of the problem by one, through a conversion from nodes $(i,\,j,\,k)$ to ``difference groups."
Let difference group $\langle l, m \rangle$ be the set of all nodes $(i,\,j,\,k)$ with $l = (j-i )~\text{mod}~N$ and $m = (k-j )~\text{mod}~N$.
For example, nodes $(1,\,2,\,3)$, $(2,\,3,\,4)$, $\cdots$, $(N, 1, 2)$ all belong to difference group $\langle 1, 1 \rangle$, and nodes $(1,\,2,\,4)$, $(2,\,3,\,5)$, $\cdots$, $(N,\,1,\,3)$ belong to difference group $\langle 1, 2 \rangle$.

\begin{theorem}
    Nodes inside a difference group are either all conflict or all non-conflict.    
    A difference group $\langle l, m \rangle$ is a conflict group exactly when $l=N$ or $m=N$ or $l+m=N$.
\end{theorem}
\begin{proof}
    This trivially follows from the definition of conflict nodes and difference groups.
\end{proof}

We can therefore divide the difference groups into conflict groups and non-conflict groups (Fig.~\ref{fig:ben}).
For example, the difference group $\langle N, 1 \rangle$ is a conflict group because it has nodes $(1,\,1,\,2)$, $\cdots$, $(N,\,N,\,1)$, which are conflict nodes since the first and second elements are identical.

Even though we have reduced the dimensionality of the problem, the transitions between conflict and non-conflict groups reflect the original transitions between conflict and non-conflict nodes.
We must therefore study the shifts between the original nodes on the 3-torus, where the walker has to walk either right ($+1$) or left ($-1$) in every dimensions. 
Each of the $ 2^3 = 8$ possible directions on the 3-torus can be translated to a unique transition between difference groups. 
Figure~\ref{fig:shifts} lists all possible shifts between nodes and their corresponding shifts between difference groups.

We can represent the difference groups as a separate $2$D-grid. 
Importantly, this is not the grid on which the quantum walker walks, as each point represents a group of nodes, and not just a single spatial position.
The transition between these groups is also different than the original movement on the 3-torus. 
In particular, as seen in Fig.~\ref{fig:shifts}(a), the walker can remain in the same difference group for multiple time steps.
This is due to the fact that it is possible to walk from one node in the same difference group to another, e.g., the walker may move from $(1,\,2,\,3)$ to $(2,\,3,\,4)$, remaining in the $\langle 1, 1\rangle$ difference group.
As Fig.~\ref{fig:shifts}(c) visually shows, each difference group in general has six ways of transitions on the 2D-grid; up, down, left, right, left-up, and right-down.
Additionally, it may remain at it's current position. 
Going down-left or right-up, however, is forbidden, as this would correspond to jumps between non-neighboring nodes in the underlying quantum walk.
These transitions between difference group are further modified by our choice of coin operator $C(x,\,y,\,z)$. 
Transitions from conflict nodes to non-conflict groups and vice versa are disallowed.

\begin{theorem}
\label{theo_same_diff_group}
    Nodes in the same difference group belong to the same subnetwork.
\end{theorem}

\begin{proof}
    Transitions from the spatial nodes $(i,\,j,\,k)$ to $(i+1,\,j+1,\,k+1)$ or $(i-1,\,j-1,\,k-1)$ do not change the difference group.
    Since nodes inside the same difference group are either all conflict or all non-conflict nodes, this transition is not a transition between conflict and non-conflict nodes, and is therefore allowed by the coin operators $C(x,\,y,\,z)$. 
    By repeatedly applying this transition, it is easy to see that we can reach all nodes within the same difference group.
    All nodes within a difference group are therefore part of the same subnetwork.
\end{proof}

Because of Theorem \ref{theo_same_diff_group}, considering the transition between difference groups is sufficient for analysing subnetworks.

First, the case where $N$ is odd in Table~\ref{tab:deduction} is proven.
\begin{theorem}\label{thm:odd}
    When $N$ is odd, there are three subnetworks. 
    Two subnetworks have $N(N-1)(N-2)/2$ non-conflict nodes, and the other one has $3N^2-2N$ conflict nodes.
\end{theorem}

\begin{proof}
   
    When the walker starts from non-conflict nodes (groups), it cannot go to conflict nodes (groups), and vice versa.
    This is why there are separate subnetworks.
    Figure~\ref{fig:7and8}(a) in the supplementary material depicts the transition between difference groups where $N=7$.    

    For general odd $N$, Fig.~\ref{fig:odd_general} shows all $N^2$ difference groups and all possible transitions.
    We will use this general graphical layout as the basis of our proof.
    Note, that the system is unfolded in a way where the difference group $\langle 1,1 \rangle$ is not next to $\langle 1, 2 \rangle$, as these difference groups are not connected by transitions.
    First, we consider the movement limit from a difference group $\langle 1, 1 \rangle$, which is a non-conflict group.
    It is necessary to consider movement in six different directions shown in Fig.~\ref{fig:shifts}.
    Regarding upward movement from $\langle 1, 1 \rangle$, it is only possible to move up to $\langle N-2, 1 \rangle$ because the next position, $\langle 0, 1 \rangle$, is a conflict group.
    Downward movement is not allowed because the position below, $\langle N-1, 1 \rangle$, is a conflict group.
    Leftward movement is also not allowed because the left adjacent difference group is the conflict group $\langle 1, N-1 \rangle$.
    Regarding the movement to the right, it is only possible to move up to $\langle 1, N-2 \rangle$ because the position on the right, $\langle 1, 0 \rangle$, is a conflict group.
    Regarding the diagonal upward movement to the left, it is only possible to move up to $\langle N-2, 4 \rangle$ because the position at the upper left, $\langle 0, 2 \rangle$, is a conflict group.
    With respect to diagonal downward movement to the right, it is only possible to move up to $\langle 4, N-2 \rangle$ because the position on the lower right, $\langle 2, 0 \rangle$, is a conflict group.
    For the above reasons, the difference group that can be moved to from $\langle 1, 1 \rangle$ is limited to the region shown in Fig.~\ref{fig:odd_general} as SN1, which is a right-angled triangular area with the hypotenuse on the lower left.
    The number of its difference groups is 
    \begin{equation}
        \frac{1}{2}(N-2+1)(N-2) = \frac{(N-1)(N-2)}{2}.
    \end{equation}

    Similarly, the range of movement from $\langle N-1, N-1 \rangle$ is limited to a right-angled triangular area with the hypotenuse pointing to the upper right.
    This forms subnetwork 2, and the number of difference groups is $(N-1)(N-2)/2$, too.

    All other conflict groups belong to subnetwork 3 because there exists a route that connects all of the conflict groups.
    The number of difference group belonging to subnetwork 3 is 
    \begin{equation}
        N^2 - \frac{(N-1)(N-2)}{2}\cdot 2=3N-2.
    \end{equation}

    The number of nodes in each subnetwork is $N$ times greater than the number of difference groups because each different groups has $N$ nodes due to Theorem \ref{theo_same_diff_group}.

\end{proof}

Second, the case where $N$ is even in Table~\ref{tab:deduction} is proven. 
\begin{theorem}\label{thm:even}
    When $N$ is even, there are nine subnetworks.
    Non-conflict nodes are divided into two subnetworks with $N(N-2)(N-4)/8$ nodes, and three subnetworks with $N^3/4-N^2/2$ nodes.
    Conflict nodes are divided into one subnetwork with $3N^2/2-2N$ nodes and three subnetworks with $N^2/2$ nodes.
\end{theorem}

\begin{proof}
    When $N$ is even, parity plays an important role.
    For a difference group $\langle m, n \rangle$, we call $(m+n)\mathrel{\mathrm{mod}}2$ its parity.
    When $N$ is even, difference groups cannot transition to other difference groups with different parity, as transitions between difference group always change $m$ and $n$ by even amounts. 
    Therefore, all difference groups are divided into four quadrants evenly; the odd-odd quadrant consisting of difference groups whose elements are both odd, the odd-even quadrant consisting of difference groups whose first element is even and second element is odd, the even-odd quadrant consisting of difference groups whose first element is even and second element is odd, and the even-even quadrant consisting of difference groups whose elements are both even.
    The walker cannot leave the quadrant it is starting in.


    For general even $N$, Fig.~\ref{fig:general_even} in the supplementary material shows all $N^2$ difference groups and all possible transitions.
    Difference groups are divided into four quadrants with $N^2/4$ difference groups each.
    The full deduction of the size and number of subnetworks is shown in the supplementary material.
    
\end{proof}

Table~\ref{tab:deduction} summarizes the sizes of the subnetworks.
This network structure analysis allows us to set appropriate condition for the initial state so that all non-conflict nodes will have non-zero observation probabilities.
A particularly parsimonious choice here is to choose one node per non-conflict subnetwork, and let the initial state should be a superposition of these nodes.
This will lead to all non-conflict nodes showing non-zero probability.
By choosing the proper starting amplitude, we can also influence the relative probabilities of different nodes and subnetworks. 
This crucial aspect allows controlling the time evolution, and thus controlling the relative selection rates of the options in the collective decision-making problem.
A systematic analysis thereof is out of the scope of the current manuscript, but a necessary step for future uses of quantum walks for multi-player decision making.



The subnetwork analysis reveals an appropriate initial state. 
In other words, the initial state needs to have probability amplitudes in all the disconnected subnetworks for all non-conflict nodes to be accessible.
As previously mentioned, there are two separate subnetworks for the case of $N=5$. By chosing a probability amplitude of $1/\sqrt{2}$ for nodes $(1,\,2,\,3)$ and $(3,\,2,\,1)$ as the intial state, we are seeding both subnetworks.
The results of our numerical simulations are shown in Figure~\ref{fig:3p_avd}(c).
It succesfully demonstrates that all non-conflict nodes have non-zero probability for the case of $N=5$.
Naturally, one could have explored the size and number of subnetworks purely numerically, e.g., the structure for $N=5$ was already apparent from our previous simulations shown in Figure~\ref{fig:3p_avd}(a) and~(b). 
However, systematically analyzing the topology of the resulting graphs makes it possible to treat the case of an arbitrary number of arms $N$, i.e., decision problems with an arbitrary number of options.

\label{section:threeplayers}


\section{Conclusion}
The purpose of this paper is to develop the framework for using quantum walks as tools in solving decision-making problems, in particular, by avoiding decision conflicts when multiple agents choose the same options.
First, we showed that a simple tensor product of two $1$D-quantum walks cannot avoid conflicts, even when considering entangled initial states. 
However, when we allow the coin operator to be entangled, i.e., when the process is better described by a single quantum walker on a $2$D torus, perfect conflict avoidance is possible by correctly choosing the coins.

The generalization of this concept to the three player case required additional considerations.
We showed the result of a ``naive" way to generalize via numerical simulation. And while decision conflicts can be completely avoided, only half of the non-conflict nodes had non-zero observation probabilities.
This is due to the fact, that the restrictions imposed on the QW via conflict-avoidance cause the network to decompose into several subnetworks, that remain functionally separate.
We derived the size and attributes of the subnetworks by reducing the problem to that of ``difference groups" in the three player case.
Numerical simulations show that, when considering these subnetworks, it is possible to arrive at probability distributions suitable for decision-making.

Our theory indicates a clear path towards further generalization for more players. However, the relative number of conflict nodes increases with the number of players, and as such, the topological structure of the subnetworks will become even more intricate.


\begin{acknowledgments}
This research was funded in part by the Japan Society for the Promotion of Science through Grant-in-Aid for Transformative Research Areas (A) (JP22H05197) and Grant-in-Aid for JSPS Fellows (JP24KJ0864). 
\end{acknowledgments}



\nocite{*}

\bibliography{bib}

\end{document}


\preprint{APS/123-QED}

\title{Suppelementary materials}


\section{Sumppelementary material for section \ref{sec:fermion}: Fermionic initial states for decision making of two players}

The following provides the proof of Theorem \ref{the:fermion} and explains why it is impossible to find a unitary matrix $W$ satisfying Eq.~\eqref{eq:w_condition}.

By using tensor products of QWs on $1D$ lattices, it is possible to create a two-player decision making system. 
However, certain conditions for the time evolution matrix $W$ need to be fulfilled to allow such a system to avoid conflicts.
In particular, we want to achieve:
\begin{equation}
    \Phi_{all}(x,\,y) = -\Phi_{all}(y,\,x),
\end{equation}

This is described by Theorem \ref{the:fermion} in the main text, which is reproduced for convenience below:
\begin{theorem}
    Let $W$ be the time evolution matrix of the QW and $U_\sigma$ be a swap operator: $(U_\sigma \Phi_{all})(x,\,y)=\Phi_{all}(y,\,x)$.
    \begin{equation}
    \begin{split}
        &W = U_\sigma W U_\sigma^{-1}, ~U_\sigma\Phi_{all}^{(0)} = -\Phi_{all}^{(0)} \\
        &\Leftrightarrow~ U_\sigma \Phi_{all}^{(t)} = -\Phi_{all}^{(t)}, \\
        &\hspace{20pt} \text{i.e., }~ \Phi_{all}^{(t)}(x,\,y) = -\Phi_{all}^{(t)}(y,\,x),~^\forall t \ge 0.
    \end{split}
    \end{equation}
\end{theorem}

Therefore, based on this theorem, if a $W$ satisfying 
\begin{equation}
\begin{split}
    &W = U_\sigma W U_\sigma^{-1}, \\
    &\text{i.e.,}~ (W\Phi_{all})(x,\,y) = (U_\sigma W U_\sigma^{-1}\Phi_{all})(x,\,y) 
\end{split}
\end{equation}
could be found, the joint probability amplitude has fermionic properties and complete conflict avoidance is possible.

First, the proof of Theorem \ref{the:fermion} follows.
\begin{proof}
    Suppose $U_\sigma \Phi_{all}^{(t)} = -\Phi_{all}^{(t)}$ holds for any $t \ge 0$.
    At some time $T$, 
    \begin{equation}
        U_\sigma \Phi_{all}^{(T)} = -\Phi_{all}^{(T)}
    \end{equation}    
    holds. At the subsequent time $T+1$,
    \begin{equation}
        U_\sigma  W \Phi_{all}^{(T)} = - W \Phi_{all}^{(T)} = W (U_\sigma \Phi_{all}^{(T)} )
    \end{equation}
    Therefore, 
    \begin{equation}
        U_\sigma  W = W U_\sigma ~\Leftrightarrow~ W = U_\sigma W U_\sigma^{-1}.
    \end{equation}
    
    Next, suppose $W = U_\sigma W U_\sigma^{-1}, ~U_\sigma\Phi_{all}^{(0)} = -\Phi_{all}^{(0)}$ holds and $U_\sigma \Phi_{all}^{(T)} = -\Phi_{all}^{(T)}$ at time $T$.
    If $U_\sigma W \Phi_{all}^{(T)} = - W \Phi_{all}^{(T)}$ holds, then the proof is done by the mathematical induction.
    Actually,
    \begin{align}
        U_\sigma W \Phi_{all}^{(T)} &= U_\sigma W (U_\sigma^{-1}U_\sigma)\Phi_{all}^{(T)} \\
        &= (U_\sigma W U_\sigma^{-1}) (U_\sigma \Phi_{all}^{(T)}) \\
        &= -W \Phi_{all}^{(T)}.
    \end{align}
    Therefore, $U_\sigma \Phi_{all}^{(T+1)} = -\Phi_{all}^{(T+1)}$ also holds.
    In other words, $\Phi_{all}^{(t)}(x,\,y) = -\Phi_{all}^{(t)}(y,\,x)$ holds for any $t \ge 0$.
\end{proof}

However, it will be shown in the following that such a $W$ cannot exist. 
The left side of Eq.~\eqref{eq:w_condition} is a vector with four elements that are probability amplitudes of inner states of node $(x,\,y)$ at time $t+1$.
Its first, second, third, and fourth element is a probability amplitude of inner state LL, LR, RL, and RR, respectively.
The inner state LL, LR, RL, and RR at node $(x,\,y)$ comes from node $(x+1,\,y+1),~(x+1,\,y-1),~(x-1,\,y+1),~(x-1,\,y-1)$, respectively.
Therefore, the left side of Eq.~\eqref{eq:w_condition} is
\begin{align}
\begin{split}
    &W\Phi_{all}^{(t)} (x,\,y) \\
    &= E_{--}(x+1,\,y+1)\Phi_{all}^{(t)} (x+1,\,y+1)  \\
    &+E_{-+}(x+1,\,y-1)\Phi_{all}^{(t)} (x+1,\,y-1) \\
    &+ E_{+-}(x-1,\,y+1)\Phi_{all}^{(t)} (x-1,\,y+1) \\
    &+ E_{++}(x-1,\,y-1)\Phi_{all}^{(t)} (x-1,\,y-1).
\end{split}
\end{align}
Here, 
\begin{align}
    E_{--}(x,\,y)=\left[\begin{array}{cccc}
        1 & 0 & 0 & 0 \\
        0 & 0 & 0 & 0 \\
        0 & 0 & 0 & 0 \\
        0 & 0 & 0 & 0 \\
    \end{array}\right]C(x,\,y),\\
    E_{-+}(x,\,y)=\left[\begin{array}{cccc}
        0 & 0 & 0 & 0 \\
        0 & 1 & 0 & 0 \\
        0 & 0 & 0 & 0 \\
        0 & 0 & 0 & 0 \\
    \end{array}\right]C(x,\,y), \\
    E_{-+}(x,\,y)=\left[\begin{array}{cccc}
        0 & 0 & 0 & 0 \\
        0 & 0 & 0 & 0 \\
        0 & 0 & 1 & 0 \\
        0 & 0 & 0 & 0 \\
    \end{array}\right]C(x,\,y),\\
    E_{--}(x,\,y)=\left[\begin{array}{cccc}
        0 & 0 & 0 & 0 \\
        0 & 0 & 0 & 0 \\
        0 & 0 & 0 & 0 \\
        0 & 0 & 0 & 1 \\
    \end{array}\right]C(x,\,y).
\end{align}
Similarly, the right side of Eq.~\eqref{eq:w_condition} is
\begin{equation}
\begin{split}
    &(U_\sigma W U_\sigma^{-1} \Phi_{all}^{(t)})(x,\,y) = (WU_\sigma^{-1} \Phi_{all}^{(t)})(y,\,x) \\
    &=E_{--}(y+1,\,x+1) U_\sigma  \Phi_{all}^{(t)} (y+1,\,x+1) \\
    &\hspace{10pt}+ E_{-+}(y+1,\,x-1) U_\sigma \Phi_{all}^{(t)} (y+1,\,x-1) \\
    &\hspace{10pt}+ E_{+-}(y-1,\,x+1) U_\sigma \Phi_{all}^{(t)} (y-1,\,x+1) \\
    &\hspace{10pt}+ E_{++}(y-1,\,x-1) U_\sigma \Phi_{all}^{(t)} (y-1,\,x-1) \\
    &=E_{--}(y+1,\,x+1)\Phi_{all}^{(t)} (x+1,\,y+1) \\
    &\hspace{10pt}+ E_{-+}(y+1,\,x-1)\Phi_{all}^{(t)} (x-1,\,y+1) \\
    &\hspace{10pt}+ E_{+-}(y-1,\,x+1)\Phi_{all}^{(t)} (x+1,\,y-1) \\
    &\hspace{10pt}+ E_{++}(y-1,\,x-1)\Phi_{all}^{(t)} (x-1,\,y-1)
\end{split}
\end{equation}
To satisfy Eq.~\eqref{eq:w_condition}, the equations below should hold for any $x$ and $y$.
\begin{equation}
    \begin{split}
        E_{--}(x,\,y) &= E_{--}(y,\,x),\\
        E_{++}(x,\,y) &= E_{++}(y,\,x) \\
    \end{split}
\end{equation}
\begin{equation}
    \begin{split}
        E_{-+}(x,\,y) &= E_{+-}(y,\,x),\\
        E_{+-}(x,\,y) &= E_{-+}(y,\,x) \label{eq:contradiction}
    \end{split}
\end{equation}
However, Eq.~\eqref{eq:contradiction} becomes 
\begin{equation}
\begin{split}
    &E_{-+}(0,\,0) = E_{+-}(0,\,0),\\
    &\text{i.e.,}~ \left[\begin{array}{cccc}
        0 & 0 & 0 & 0 \\
        0 & 0 & 0 & 0 \\
        0 & 0 & 1 & 0 \\
        0 & 0 & 0 & 0 \\
    \end{array}\right]C(0,\,0)=\left[\begin{array}{cccc}
        0 & 0 & 0 & 0 \\
        0 & 1 & 0 & 0 \\
        0 & 0 & 0 & 0 \\
        0 & 0 & 0 & 0 \\
    \end{array}\right]C(0,\,0)
\end{split}
\end{equation}
when $x=y=0$.
It means that the coin at node $(0,\,0)$, $C(0,\,0)$, is not unitary because the second and third rows are the same.
As a result, no unitary time evolution matrices can meet Eq.~\eqref{eq:w_condition} and fulfill Theorem~\ref{the:fermion}.

\section{Supplementary material of section \ref{section:2d}: If coins at conflict nodes are entangled}

The main text has discussed how changing the coins at the border nodes can influence the transitions of the quantum walker on a 2D-grid.
Here, we want to discuss the related idea of changing the coins at the conflict nodes as shown in Fig.~\ref{fig:2p_leaveimmid} (red nodes, top right).
Naively, this may seem the more natural approach, but as we show below it will turn out to be insufficient for conflict-avoidance.
For this approach in particular, we want the walker to transition to non-conflict nodes at the following time step if it is at conflict nodes at the current time step.
Thus, we need to choose the coin matrix such that the inner states coming from border nodes toward the conflict node are transformed into inner states which will go back toward border nodes in the next step.
These inner states are represented by black arrows in Fig.~\ref{fig:2p_leaveimmid}(a).
The walker entering a conflict node thus always goes back to  border nodes at the subsequent time step.
In other words, we disallow walkers to ``stay" on conflict nodes, or to transition from one conflict node to the next.
Due to the constraint that coins need to be unitary, this in turn also implies that the inner states connecting between conflict nodes (red arrows in Fig.~\ref{fig:2p_leaveimmid}(a)) will also be turned into inner states that go back to conflict nodes.
As a result, the walker will continue onward to the next conflict node if it came from another conflict node.
However, as long as the walker does not start on any conflict node, it will never access this subnetwork of connected inner states on the conflict nodes. 
It is therefore important to avoid a wrong initial state.

For $N=11$, the simulated average probability distribution of nodes is displayed in Fig.~\ref{fig:2p_leaveimmid}.
Conflict nodes had the coin 
\begin{equation}
    \frac{1}{\sqrt{2}}\left[ \begin{array}{cccc}
        1 & 0 & 0 & 1 \\
        0 & 1 & 1 & 0 \\
        0 & 1 & -1 & 0 \\
        1 & 0 & 0 & -1
    \end{array} \right],
\end{equation}
and other nodes had the coin 
\begin{equation}
    \frac{1}{\sqrt{2}}\left[ \begin{array}{cc}
        1 & 1 \\
        1 & -1
    \end{array} \right] \otimes \frac{1}{\sqrt{2}}\left[ \begin{array}{cc}
        1 & 1 \\
        1 & -1
    \end{array} \right].
\end{equation}
These coins are unitary.
The average probability distribution of conflict nodes along the main diagonal (see the arrow in Fig.~\ref{fig:2p_leaveimmid}(b)) is less than the gray plane of the uniform distribution, suggesting that the conflict probability is suppressed but not eliminated.
Thus, while this approach overall requires fewer modified coins, it also does not eliminate decision conflicts, i.e., the probability to observe the quantum walker at nodes associated with both players choosing the same option is non-zero.

\begin{figure}[t]
   \centering
    \includegraphics[width=0.4\textwidth]{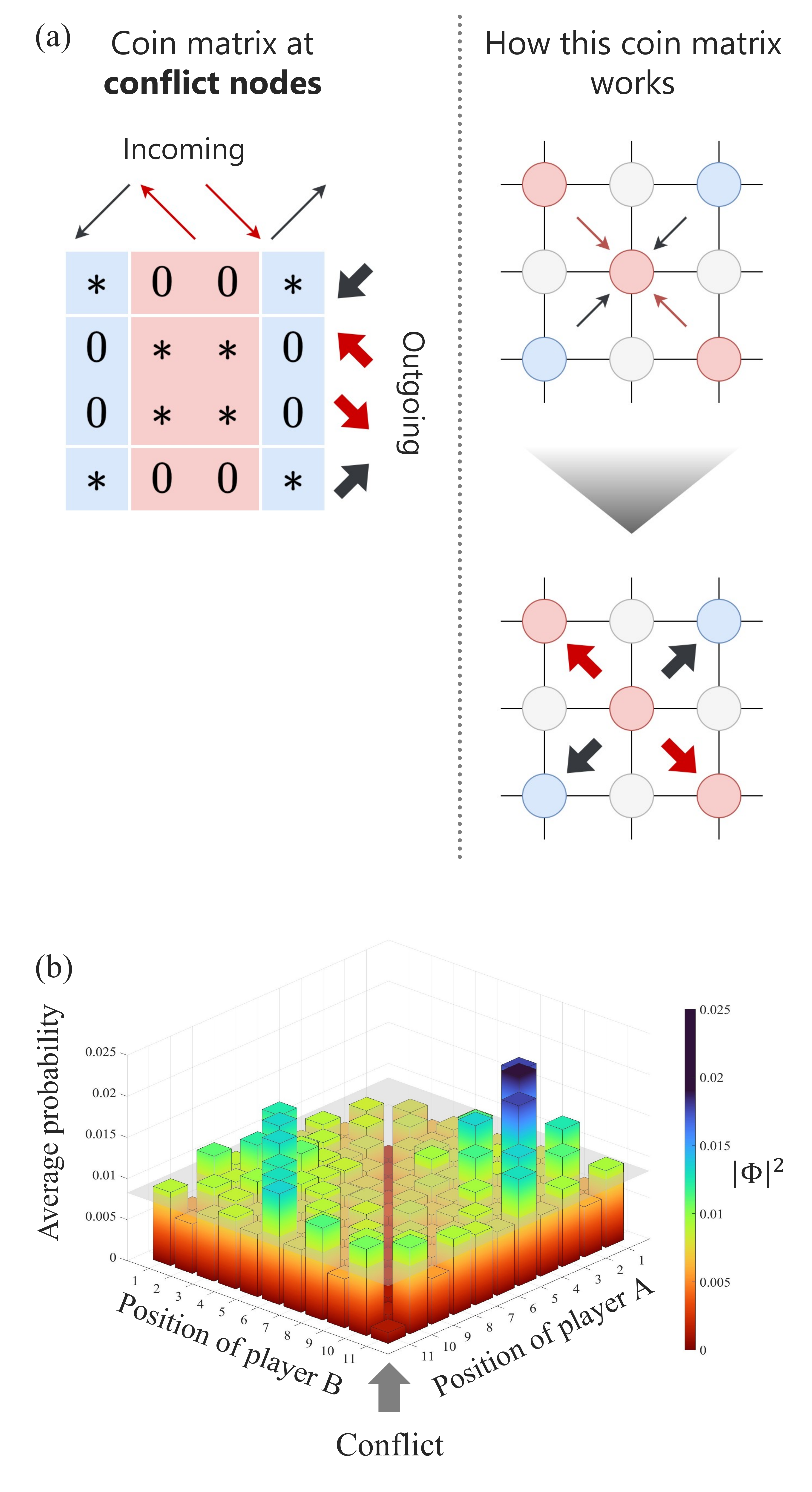}
   \caption{Sketch of required choice of coin operators (see Eq.~\eqref{eq:2d_coin}) to avoid decision conflicts for 2-player decision making. Coin operators $C(x,\,y)$ are represented as a $4 \times 4$ matrix. Each node $(x,\,y)$ has four inner states, shown by arrows for indicating incoming and outgoing probability amplitudes. (a) The coin matrices at conflict nodes $(i,\,i)$ are entangled. Red color indicates inner states from and to directions of other conflict nodes. Matrix elements shown as $*$ may take any allowable value under the constraint that the coin matrix must be unitary. (b) Average probability density of simulated quantum walk for $N=11$ using special coins at conflict nodes. Note that conflict, while reduced in amplitude, is not perfectly avoided.}
   \label{fig:2p_leaveimmid}
\end{figure}

\section{Supplementary materials for  section \ref{section:threeplayers}}

\begin{figure*}[p]
   \centering
    \includegraphics[width=0.8\textwidth]{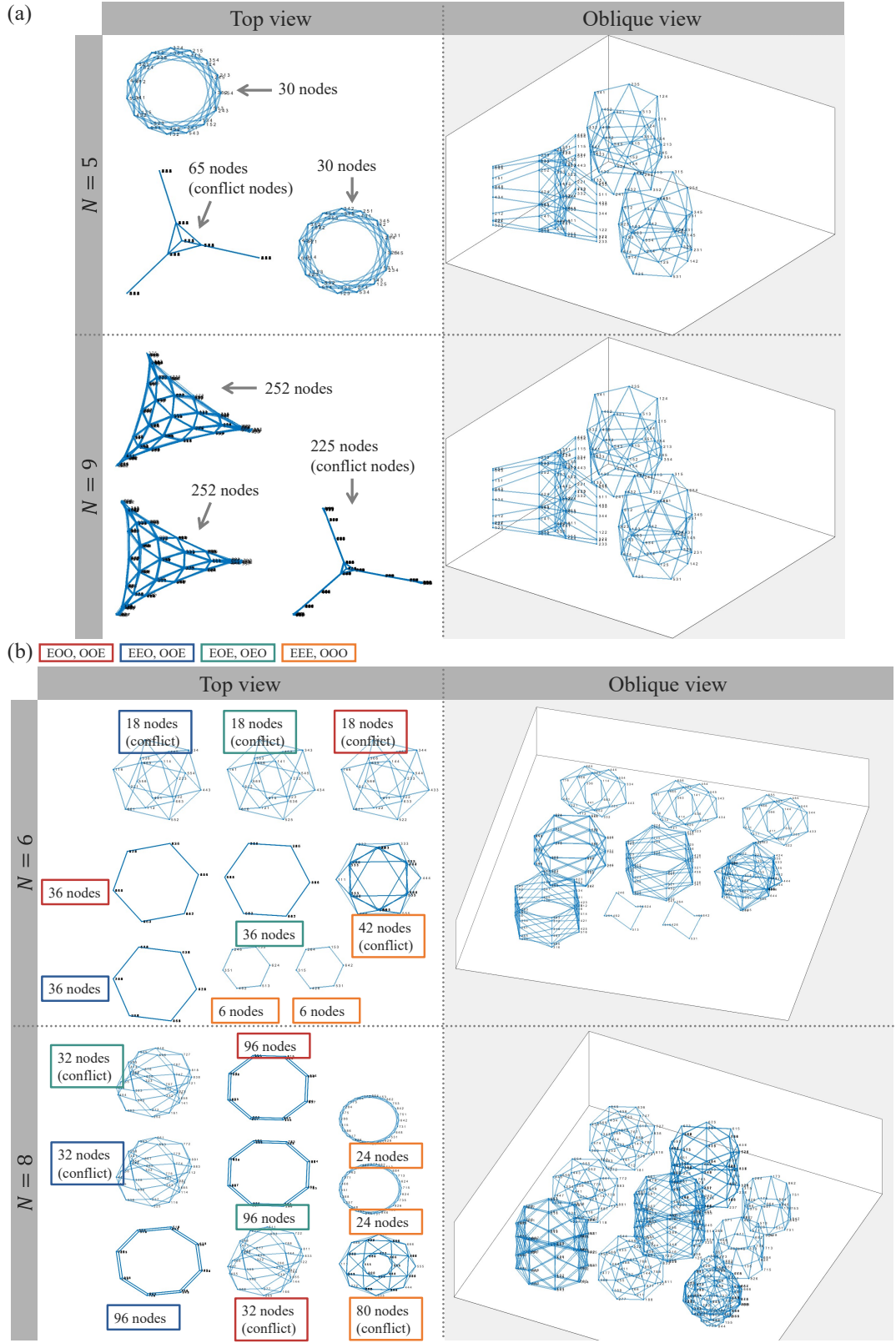}
   \caption{Examples of subnetwork structures when (a) $N$ is odd ($N=5,\,9$) and (b) $N$ is even ($N=6,\,8$).}
   \label{fig:networks}
\end{figure*}

\subsection{Analysis of the network structure}
This section describes how we have conjectured from the simulation results those elements of Table~\ref{tab:deduction} that we cannot prove formally.

At the beginning, it is only clear that nodes are divided into some subnetworks from the results displayed in Fig.~\ref{fig:3p_avd}.
To deduce the network structure for general $N$, the network structure for $N=5,6,8,9$ is examined in simulations.
First, to consider the network structure for general odd $N$, numerical simulations for $N=5$ and $9$ are done (Fig.~\ref{fig:networks}(a)).
In both cases, nodes are divided into three subnetworks.
All conflict nodes belong to one, and non-conflict nodes are divided evenly into two subnetworks.
Since connectivity in general increases with $N$, not decreases, one can conjecture that there are exactly three subnetworks when $N$ is odd.
One subnetwork has all conflict nodes, whose number is $3N^2-2N$, and non-conflict nodes, whose number is $N(N-1)(N-2)$, are divided equally into the other two subnetworks (Table~\ref{tab:deduction}).

Second, simulations for $N=6$ and $8$ are conducted to deduce the network structure for general even $N$ (Fig.~\ref{fig:networks}(b)).
When $N$ is even, the parity divides nodes into four groups first because each element changes from even to odd or from odd to even at every time step.
The four groups are (EEE, OOO), (EOO, OOE), (EEO, OOE), and (EOE, OEO), where E and O denote even and odd, respectively, and patterns written in the same parenthesis can transition into each other, while others cannot not.
For example, suppose the walker starts at node $(1,2,3)$.
The start node belongs to OEO nodes because $1$ is odd, $2$ is even, and $3$ is odd.
The walker will be at EOE node at the first time step, and it will be at OEO node at the second time step, and so on.
However, no matter what, the walker cannot move arrive at any node that does not fit either the EOE (even-odd-even) or OEO (odd-even-odd) pattern.
Fig.~\ref{fig:networks}(b) indicates there are nine subnetworks in both cases.
The groups (EOO, OOE), (EEO, OOE), and (EOE, OEO) are divided into two subnetworks respectively, and the group (EEE, OOO) is divided into three subnetworks. 
We do not have a formal proof for why some of the groups split further into two subnetworks, but as we can see visually from Fig.~\ref{fig:general_even}, the reason for that is that these quadrants contain conflict nodes, which are inaccessible due to our choice of coin operators.
As a result, the number of nodes belonging to each subnetwork is summarized in Table~\ref{tab:deduction}.

\begin{figure*}[p]
   \centering
    \includegraphics[width=0.8\textwidth]{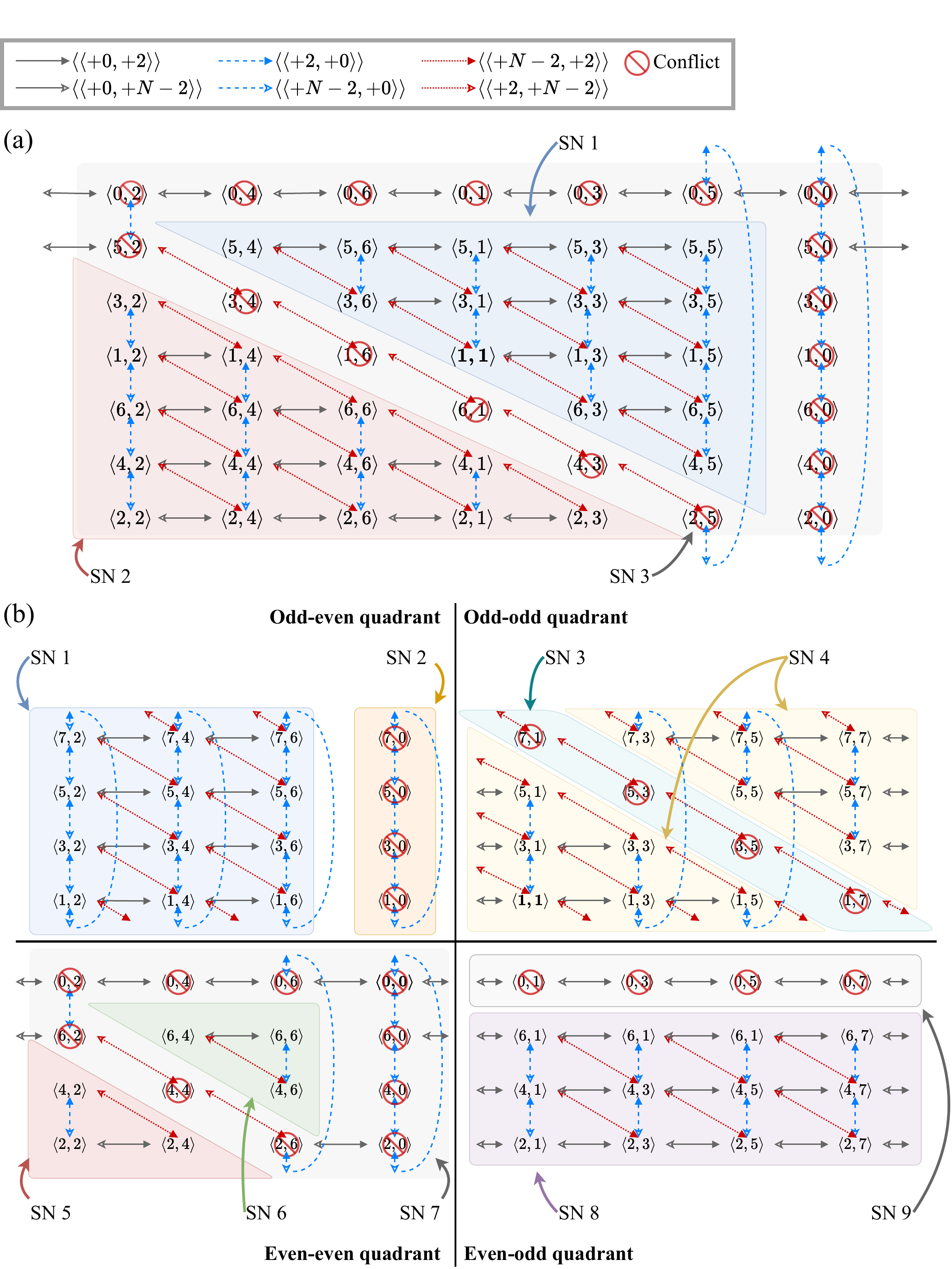}
   \caption{Transition between difference groups when (a) $N=7$ and (b) $N=8$. The brackets indicating difference groups are omitted in this figure. SN stands for the subnetwork.}
   \label{fig:7and8}
\end{figure*}

\begin{figure*}[t]
   \centering
    \includegraphics[width=0.9\textwidth]{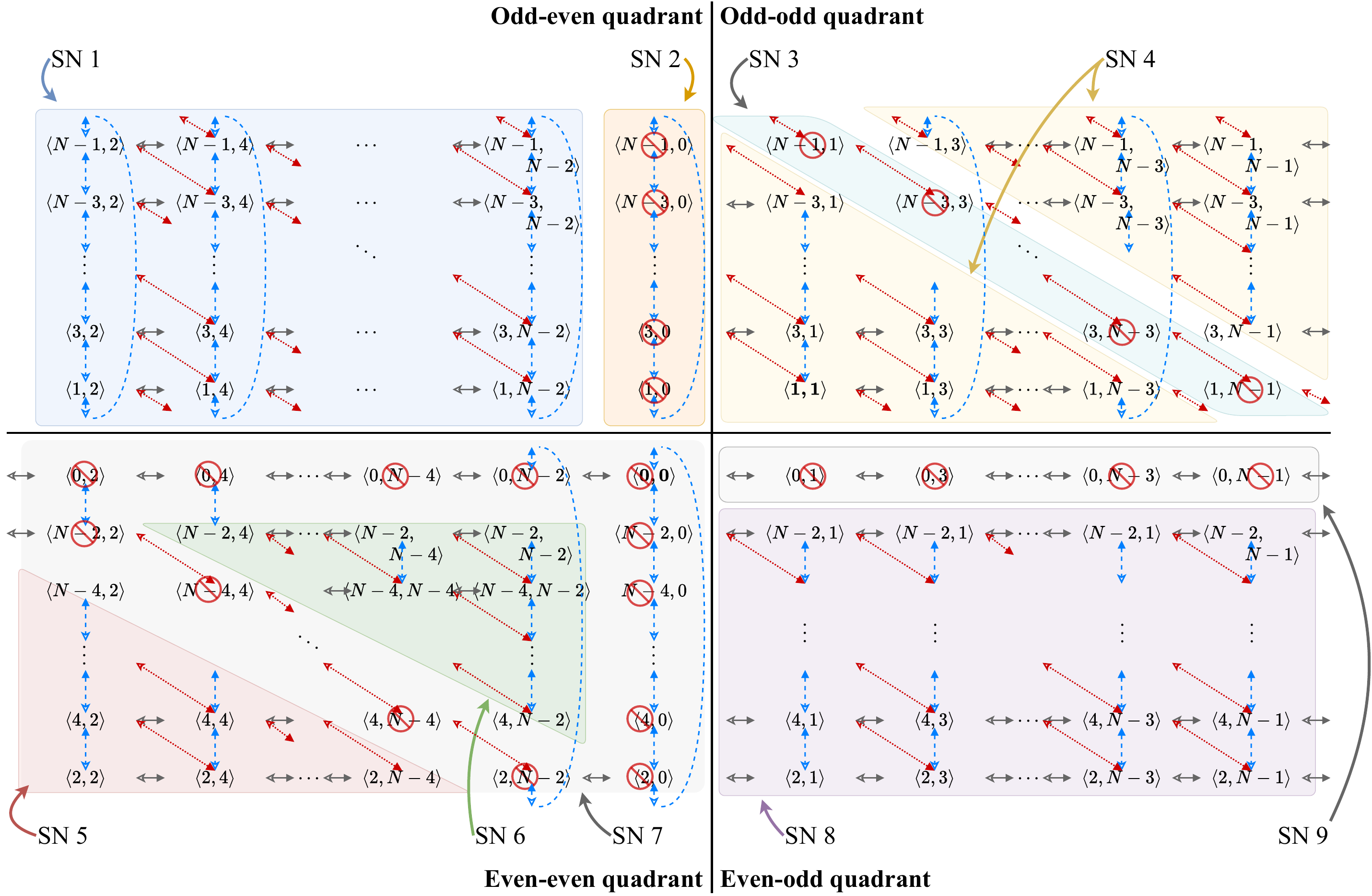}
   \caption{Transitions between difference group when $N$ is any even number. SN stands for the subnetwork.Note that the 2D grid is ordered in a way where the difference group $\langle 1, 1 \rangle$ is not next to $\langle 1, 2 \rangle$, as these difference groups are not connected by transitions.}
   \label{fig:general_even}
\end{figure*}

\subsection{Proof for the number of subnetworks for even $N$}
This section provides the full proof of Theorem~\ref{thm:even}. It proves the sizes and types of some of the subnetworks.

\begin{proof}
As stated in the main manuscript, the analysis of the number and size of the subnetworks for an even number of options $N$ strongly depends on the parity of the nodes.

First, we focused on the odd-even quadrant.
The range of movement from $\langle 1, 2\rangle$, which is a non-conflict group, is defined as subnetwork 1.
First, movement to the left is not possible because the left adjacent group $\langle 1, 0 \rangle$ is a conflict group.
Moving to the right leads to a dead end at $\langle 1, N-2 \rangle$, since the next to the right, $\langle 1, 0 \rangle$, is also a conflict group.
Moving either up or down passes only through non-conflict groups, eventually circling back to $\langle 1, 2\rangle$.
Therefore, the range of movement from $\langle 1, 2\rangle$ is limited to the rectangular area shown as SN1 in Fig.~\ref{fig:general_even}.
It has $N/2(N/2-1)$ difference groups, and therefore $N^2(N-2)/4$ nodes.

The range of movement from $\langle 1, 0\rangle$, which is a conflict group, is defined as subnetwork 2.
Movement in the left and right directions is not possible because the adjacent groups, $\langle1, N-2\rangle$ and $\langle1, 2\rangle$ are non-conflict groups.
Similarly, diagonal movement is also not possible, as the upper-left diagonal group $\langle 3, N-2\rangle$ and the lower-right diagonal group $\langle N-1, 2\rangle$ are non-conflict groups.
Therefore, considering only vertical movement,  eventually looping back to $\langle 1, 0\rangle$.
This is shown as SN2 in Fig.~\ref{fig:general_even}.
It has $N/2$ conflict groups, and $N^2/2$ nodes.

Second, we focused on the odd-odd quadrant.
The range of movement from $\langle 1, N-1\rangle$, which is a conflict group, is defined as subnetwork 3.
Vertical movement is not possible because the groups above, $\langle 3, N-1\rangle$, and below, $\langle N-1, N-1\rangle$, are non-conflict groups.
Similarly, horizontal movement is also not possible because the positions to the left, $\langle 1, N-3\rangle$, and to the right, $\langle 1, 1\rangle$, are non-conflict groups.
Therefore, considering only diagonal movement, one eventually loops back to $\langle 1, N-1\rangle$.
This is shown as SN3 in Fig.~\ref{fig:general_even}.
It has $N/2$ conflict groups, and $N^2/2$ nodes.

The range of movement from $\langle1, 1\rangle$ is defined as subnetwork 4.
Since the group below, $\langle N-1, 1\rangle$, is a conflict group, downward movement is not possible. Moving upward leads to a dead end at $\langle N-3, 1\rangle$ for the same reason.
Since the group to the left, $\langle 1, N-1\rangle$, is a conflict group, leftward movement is not possible. Moving to the right leads to a dead end at $\langle 1, N-3\rangle$ for the same reason.
Moving diagonally to the upper left connects to $\langle N-1, N-1\rangle$. From $\langle N-1, N-1\rangle$, moving left leads to a dead end at $\langle N-1, 3\rangle$, and moving downward leads to a dead end at $\langle 3, N-1\rangle$.
Therefore, the range of movement from $\langle 1, 1\rangle$ is limited to the area shown as SN4 in  Fig.~\ref{fig:general_even}.
The subnetwork 4 has $N/2(N/2-1)$ non-conflict groups and therefore $N^2(N-2)/4$ nodes.

Next, we focused on the even-even quadrant.
The range of movement from $\langle 2,2 \rangle$, which is a non-conflict group, is defined as subnetwork 5.
Regarding upward movement from $\langle 2,2 \rangle$, it is only possible to move up to $\langle N-4, 2 \rangle$ because the next position, $\langle N-2,2 \rangle$, is a conflict group.
The downward movement is not allowed because the position below, $\langle 0, 2 \rangle$, is a conflict group.
Leftward movement is also not allowed because the position to the left $\langle 2, 0 \rangle$ is a conflict group 
Regarding rightward movement, it is only possible to move up to $\langle 2, N-4 \rangle$ because the position to the right, $\langle 2, N-2 \rangle$, is a conflict group.
For the above reasons, the difference group that can be moved to from $\langle 2,2 \rangle$ is limited to the region shown in Fig.~\ref{fig:general_even} as SN5.

Similarly, the range of movement from $\langle N-2, N-2 \rangle$ is limited to a right-angled triangular area shown as SN6.
Subnetworks 5 and 6 have $(N/2-1)(N/2-2)/2 = (N-2)(N-4)/8$ non-conflict groups and $N(N-2)(N-4)/8$ nodes.

The other conflict groups in the even-even quadrant belong to subnetwork 7 because there exists a route that connects all of the conflict groups.
It has $N^2/4-(N/2-1)(N/2-2)=(3N-4)/2$ conflict nodes, and the number of its nodes is $N(3N-4)/2$.

Finally, we focused on the even-odd quadrant.
The range of movement from $\langle2, 1\rangle$, which is a non-conflict group, is defined as subnetwork 8.
Since the position below, $\langle 0, 1\rangle$, is a conflict group, downward movement is not possible. Moving upward leads to a dead end at $\langle N-2, 1\rangle$ for the same reason.
When moving left or right, all positions are non-conflict groups, eventually circling back to $\langle2, 1\rangle$.
Therefore, the range of movement from $\langle2, 1\rangle$ is limited to the rectangular area shown as SN8 in Fig.~\ref{fig:general_even}.
It has $N/2(N/2-1)$ non-conflict groups and therefore $N^2(N-2)/4$ nodes.
    
The range of movement from $\langle 0, 1\rangle$, which is a conflict group, is defined as subnetwork 9.
Vertical movement is not possible because the positions above, $\langle 2, 1\rangle$, and below, $\langle N-2, 1\rangle$, are non-conflict groups.
Similarly, diagonal movement is also not possible because the upper-left position $\langle2, N-1\rangle$ and the lower-right position $\langle N-2, 1\rangle$ are non-conflict groups.
Therefore, considering only horizontal movement, one eventually loops back to $\langle 0, 1\rangle$.
This is shown as SN9 in Fig.~\ref{fig:general_even}.
It has $N/2$ conflict groups, and $N^2/2$ nodes.
\end{proof}

\begin{figure*}[t]
   \centering
    \includegraphics[width=0.8\textwidth]{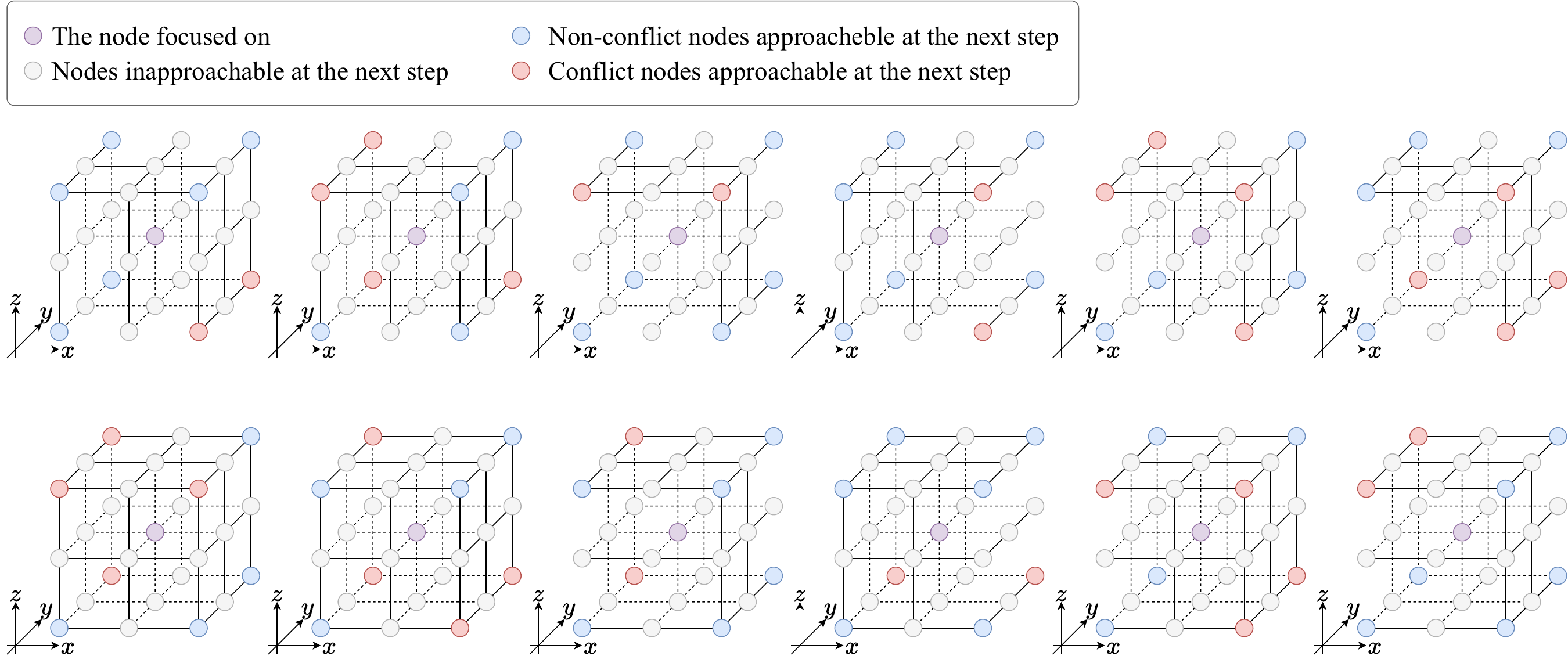}
   \caption{Twelve coin groups.}
   \label{coingroup}
\end{figure*}

\subsection{Coin group}

This section introduces a notion omitted from the main text, which is the ``coin group".

We will refer to a group of nodes with the same coin pattern, i.e., the same location relative to the conflict node as coin groups.
There are twelve coin groups shown in Fig.~\ref{coingroup}.

The coin group encompasses the difference group because of the following theorem.
\begin{theorem}
    Nodes in the same difference group belong to the same coin group.
\end{theorem}
\begin{proof}
    In the following, remind that $1-1$ and $N+1$ are treated as $N$ and $1$, respectively, as introduced for cycles in Sec.~\ref{subsec:qw1d}.
    The difference group of node $(x,\,y,\,z)$ has also nodes $(x+1,\,y+1,\,z+1)$, $\cdots$, $(x+N,\,y+N,\,z+N)$.
    Consider the coin matrix of node $(x,\,y,\,z)$.
    Each row and column corresponds to node $(x+i,\,y+j,\,z+k)$, where $i,\,j,\,k \in \{ +1,\,-1 \}$. 
    When node $(x+i, y+j, z+k)$ is a conflict (non-conflict) node, node $\bigl( x+i+2, \, y+j+2, \,  z+k+2 \bigr)$ is also a conflict (non-conflict) node.
    Coin matrix of node $(x,\,y,\,z)$ and $\bigl(x+1,\,y+1,\,z+1\bigr)$ have zero elements in the same positions since elements of the coin matrix are non-zero only when both node $(x,\,y,\,z)$ and node $(x+i, y+j, z+k)$ are conflict or non-conflict nodes.
    Therefore, nodes $(x,\,y,\,z)$, $(x+1,\,y+1,\,z+1)$, $\cdots$, $(x+N,\,y+N,\,z+N)$ belong to the same coin group.
\end{proof}
